\documentclass[11pt]{article}

\usepackage{amsthm, amssymb, amsmath, amsfonts, url, enumerate, mathtools, graphics}

\newcommand{\dst}{\displaystyle}

\newtheorem{theorem}{Theorem}[section]

\newtheorem{claim}[theorem]{Claim}
\newtheorem{construction}[theorem]{Construction}
\newtheorem{proposition}[theorem]{Proposition}

\newtheorem{corollary}[theorem]{Corollary}
\newtheorem{lemma}[theorem]{Lemma}

\newtheorem{remark}[theorem]{Remark}
\theoremstyle{definition}
\newtheorem{definition}[theorem]{Definition}

\newenvironment{prf}{\noindent{\bf Proof.~}}{\(\qed\)}
\newcommand{\BPF}{\begin{prf}} 
\newcommand {\EPF}{\end{prf}}

\newcommand{\capac}{\textnormal{cap}}
\newcommand{\Det}{\textnormal{Det}}

\newcommand{\ds}{\textnormal{ds}}
\newcommand{\tensor}{\otimes}
\newcommand{\expon}{\textnormal{exp}}

\newcommand{\tr}{\textnormal{tr}}
\newcommand{\rk}{\textnormal{rank}}
\newcommand{\Rk}{\textnormal{Rank}}

\def\poly{\textnormal{poly}}

\newcommand{\BL}{\textnormal{BL}}
\newcommand{\dm}{\textnormal{dim}}

\def\Q{{\mathbb{Q}}}
\def\Z{{\mathbb{Z}}}
\def\N{{\mathbb{N}}}
\def\R{{\mathbb{R}}}

\def\C{{\mathbb{C}}}

\def\cM{\mathcal M}

\def\bf{{\mathbf f}}

\def\bv{{\mathbf v}}
\def\bx{{\mathbf x}}
\def\by{{\mathbf y}}

\def\bv{{\mathbf v}}
\def\bx{{\mathbf x}}
\def\by{{\mathbf y}}

\newcommand{\wt}[1]{\widetilde{#1}}

\newcommand{\Matrv}{\mathcal{M}_{\mathbf{v}}}
\newcommand{\Matrw}{\mathcal{M}_{\mathbf{w}}}

\def\then{\Rightarrow}

\def\iff{\Leftrightarrow}
\def\to{\rightarrow}

\newcommand{\eps}{\epsilon}

\usepackage[margin=0.75in]{geometry}

\usepackage{verbatim}
\usepackage{enumitem}

\usepackage{graphicx}
\usepackage{epsfig}
\usepackage{float}
\usepackage[all]{xy}
\usepackage{color}

\usepackage{fullpage}
\usepackage{appendix}
\usepackage{mathtools}

\floatstyle{boxed}
\newfloat{Protocol}{h}{pro}

\floatstyle{boxed}
\newfloat{Algorithm}{h}{pro}

\usepackage{tikz}
\usetikzlibrary{arrows}
\usepackage{float}

\usepackage{hyperref}

\hypersetup{
    colorlinks,
    citecolor=blue,
    filecolor=blue,
    linkcolor=blue,
    urlcolor=blue
}

\begin{document}

\pagenumbering{gobble}

\title{Algorithmic and optimization aspects of Brascamp-Lieb inequalities, via Operator Scaling}

\author{
Ankit Garg \thanks{Microsoft Research New England, email: garga@microsoft.com. Part of research was done when the author was a student at Princeton University. Research supported by Mark Braverman's NSF grant CCF-1149888, Simons Collaboration on Algorithms and Geometry, Simons Fellowship in Theoretical Computer Science and Siebel Scholarship.}
\and
Leonid Gurvits\thanks{Department of Computer Science, The City College of New York, email: l.n.gurvits@gmail.com}
\and
Rafael Oliveira \thanks{Department of Computer Science, Princeton University, email: rmo@cs.princeton.edu. 
Research supported by NSF CAREER award DMS-1451191, NSF grant CCF-1523816 and Siebel Scholarship.}
\and
Avi Wigderson \thanks{Institute for Advanced Study, Princeton, email: avi@math.ias.edu. Research supported by NSF grant
CCF-1412958.}
}

\maketitle

\begin{abstract}

The celebrated Brascamp-Lieb (BL) inequalities~\cite{BL76,Lieb90}, and their reverse form of Barthe~\cite{Barthe2}, are an 
important mathematical tool, unifying and generalizing numerous inequalities in analysis, convex geometry 
and information theory, with many used in computer science. While their structural theory is very well understood, far less 
is known about computing their main parameters (which we later define below). Prior to this work, the best known 
algorithms for any of these optimization tasks required at least exponential time.
In this work, we give polynomial time algorithms to compute:  
(1) Feasibility of {\em \textnormal{BL}-datum}, 
(2) Optimal {\em \textnormal{BL}-constant}, 
(3) Weak separation oracle for {\em \textnormal{BL}-polytopes}.

What is particularly exciting about this progress, beyond the better understanding of BL-inequalities, is that the objects 
above naturally encode rich families of optimization problems which had no prior efficient algorithms. In particular, 
the BL-constants (which we efficiently compute) are solutions to  {\em non-convex} optimization problems, and the 
BL-polytopes (for which we provide efficient membership and separation oracles)  are linear programs with 
{\em exponentially many} facets. Thus we hope that new combinatorial optimization problems can be solved via 
reductions to the ones above, and make modest initial steps in exploring this possibility. 

Our algorithms are obtained by a simple efficient reduction of a given BL-datum to an instance of the 
{\em Operator Scaling} problem defined by~\cite{gurvits2004}. To obtain the results above, we utilize the two 
(very recent and different) algorithms for the operator scaling problem~\cite{GGOW, IQS15b}.  
Our reduction implies algorithmic versions of many 
of the known structural results on BL-inequalities, and in some cases provide proofs that are different or simpler than 
existing ones. Further, the analytic properties of the~\cite{GGOW} algorithm 
provide new, effective bounds on the magnitude and continuity of BL-constants; prior work relied on compactness, and 
thus provided no bounds.  

On a higher level, our application of operator scaling algorithm to BL-inequalities further connects analysis and optimization with the diverse mathematical areas used so far to motivate and solve the operator scaling problem, which include commutative invariant theory, non-commutative algebra, computational complexity and quantum information theory.

\end{abstract}

\newpage

\tableofcontents

\thispagestyle{empty}

\pagenumbering{arabic}

\clearpage
\setcounter{page}{1}

%
%
%



\newpage

\section{Introduction}

This long introduction summarizes the main contributions of the paper and the intuitions behind many of the definitions, results and proofs. We start with the main object of this paper: introducing the Brascamp-Lieb inequalities, and  surveying known structural results and our new algorithmic results. We then describe how to instantiate the operator scaling algorithm of~\cite{GGOW} directly in the context of BL-inequalities, in a way that does not require previous familiarity with it. The properties of this algorithm lead to further structural results, as well as highlight it as a provably efficient instance of the general {\em alternate minimization} \footnote{sometimes also called alternating minimization.} heuristic. Finally, we discuss several  known and potential applications of Brascamp-Lieb inequalities in computer science and optimization, which further motivate this and future studies.

\subsection{The Brascamp-Lieb inequalities: basic notions}

Many important inequalities, including H\"older's inequality, Loomis-Whitney inequality, Young's convolution inequality, hypercontractivity inequalities and many others are all special cases of the extremely general Brascamp-Lieb inequalities, introduced by these authors in~\cite{BL76,Lieb90}. Yet many others, including Prekopa-Lindler inequality, versions of Brunn-Minkowski  inequality and others are special cases of Barthe's reverse form of Brascamp-Lieb~\cite{Barthe2}. We discuss below only the original form, and only for Euclidean spaces. The notation we use is taken mostly from~\cite{BCCT}, which is an excellent source both for background and motivation on this topic, as well as the state-of-art on the basic questions in this field. As is common, we often use BL as abbreviation for Brascamp-Lieb.

A {\em Brascamp-Lieb datum} is specified by two $m$-tuples: one of linear transformations 
\linebreak $\mathbf{B} = (B_1, B_2, \dots, B_m)$, with $B_j:\R^n \rightarrow \R^{n_j}$ 
(along with $(n, n_1,\ldots,n_m)$, the vector of {\em dimensions}), and another of 
non-negative reals $\mathbf{p}=(p_1, p_2,\dots , p_m)$ (which is a vector of {\em exponents}). 
This combined information is denoted by $(\mathbf{B,p})$. 
A {\em Brascamp-Lieb inequality}\footnote{While this is the common name of the inequality below found in the literature, we note that only a restricted form appears in the original~\cite{BL76} paper, and this general form appeared only in the later paper of Lieb~\cite{Lieb90}.} with the datum above states that for {\em every} tuple of $m$ 
non-negative functions, $f=(f_1, f_2, \dots, f_m)$, which are integrable according to the Lebesgue 
measure, we have the following inequality.
$$  \int_{x\in \R^n} \prod_{j=1}^m (f_j(B_j x))^{p_j} dx \, \leq \, 
C \,\prod_{j=1}^m  \left(\int_{x_j \in \R^{n_j}}  f_j(x_j) dx_j \right)^{p_j} $$
where $C \in (0, \infty ]$ is independent of the functions $f$. Of course, we really get an interesting 
inequality if $C$ is finite. In that case, we call the datum $(\mathbf{B,p})$  {\em feasible}, and we 
denote smallest $C$ for which this inequality holds by $\BL(\mathbf{B,p})$, called the {\em \textnormal{BL}-constant}.

Many familiar special cases of this inequality are listed in the introduction of~\cite{BCCT}; we describe only a very simple one here, which hopefully makes concrete and intuitive the complicated formal expression above and its constituents. A special case of the Loomis-Whitney inequality asserts that {\em the volume of  every measurable set $S$ in the unit cube is at most the square root of of the product of the areas of its projections on the three coordinate planes}. In this theorem, the linear transformations $B_j:\R^3 \rightarrow \R^2$ are defined by the simple projections $B_1(x,y,z) = (y,z)$, $B_2(x,y,z) = (x,z)$ and $B_3(x,y,z) = (x,y)$,  the functions $f_j$ are the indicator functions of these three projections of the set $S$, and the exponent vector is $p=(\frac12,\frac12,\frac12)$. The corresponding BL-constant in this case happens to be 1 (this case is an important situation which will be important later). The reader may be familiar with an entropy\footnote{Appropriately defined for these continuous variables.} version of the inequality above, namely that $H(X,Y,Z) \leq \frac12 (H(X,Y)+H(Y,Z)+H(X,Z))$ for every random variable $(X,Y,Z)$. Indeed, BL-inequalities may be viewed in general as entropy subadditivity inequalities, as developed in~\cite{Carlen09, Liu16, LiuCCV17}.

There has been extensive work on precisely understanding (and computing) when a given datum 
$(\mathbf{B,p})$ is feasible, and if it is, determining the 
BL-constant $\BL(\mathbf{B,p})$. Clean characterizations exist for both questions, which clarify that they are 
both decidable (have a finite algorithm). Let us overview these characterizations and then turn to the computational complexity 
of the existing algorithms and our new ones. The first, regarding feasibility, is from the paper~\cite{BCCT}, 
and the second, on the optimal constant, is from~\cite{Lieb90}. Before doing so, we make two general comments, one about the ``reverse BL-inequalities'' and the second regarding how we measure input size when analyzing algorithms.

\paragraph{\textit{Reverse BL inequalities}}

This informal comment simply clarifies that all results in this paper stated for BL-inequalities hold for their reverse form.
In~\cite{Barthe2}, Barthe introduced a reverse form of the Brascamp-Lieb inequalities, sometimes called RBL inequalities, which turn out to generalize some known inequalities not captured by the the original BL inequalities. These RBL inequalities take the {\em same} data as standard BL inequalities, and have an optimal RBL constant associated with each. All questions raised above for the BL inequalities, like feasibility, computation of the constant and its continuity properties are relevant in this reverse setting. However, they are actually equivalent to the original ones for a simple reason: Barthe~\cite{Barthe2} proved that for any feasible datum the optimal BL constant and the optimal RBL constant multiply to $1$, and when the BL datum is infeasible the RBL constant is $0$. In short, these two optimal constants determine each other. Thus all structural results above translate to the reverse setting, and so do all our algorithmic results.

\paragraph{\textit{Input size conventions}}

Before we start, let us formalize the input conventions to all algorithms, and the parameters we 
use to measure their complexity. The input to all algorithms will be a BL datum $(\mathbf{B,p})$, and the ``size" 
of each part will be measured differently. The entries of the matrices $B_j$ will be rational numbers, 
given in binary. We will let $b=b(\mathbf{B})$ denote their total binary length (note that in particular $b \geq nm$, and so lower bounds the ``combinatorial size'' of the input). 
The vector $\mathbf{p}$ will be given as a sequence of rationals with a common denominator, namely 
$p_j = c_j/d$ with $c_j, d$ integers. We use this convention for two reasons. First, our algorithms will use 
this integer representation, and their complexity will depend on $d=d(p)$ (while in many cases $d$ is only 
polynomial in the dimensions $n,m$ of the problem, it can definitely be as large as exponential in other cases). 
Second, in the context of {\em quiver representations}, which is very relevant to this study as we shall 
later see\footnote{With hindsight, our reduction from BL data to operator scaling data may be viewed as a special case of the reduction of Derksen and Makam~\cite{DM15} of general quivers to the Left-Right quiver. More on this in Section \ref{sec:bl-to-cap}.}, it is natural to use these integer ``weights'' $c_j$ and $d$. 

Summarizing, the two size parameters for a BL datum $(\mathbf{B,p})$ will be $b$ and $d$ as above.

\subsection{The Brascamp-Lieb inequalities: known and new results}

\paragraph{\textit{Testing feasibility} (\textit{and more})}

The following theorem of Bennett et al~\cite{BCCT,BCCT2} precisely characterizes when a given BL datum is feasible. This work will provide a different proof of this important theorem (see Corollary~\ref{cor:BCCT}).

\begin{theorem}[\cite{BCCT}]\label{feasibility}
The datum $(\mathbf{B,p})$ is feasible if and only if the following conditions hold:
\begin{enumerate}
\item $n= \sum_j p_j n_j.$
\item $\dim (V) \leq \sum_j p_j  \dim(B_jV)$ \,\, holds for all subspaces $V$ of $\R^n$.
\end{enumerate}
\end{theorem}


Note that these are simply linear conditions on the exponent vector $\mathbf{p}$, albeit defined 
by infinitely many subspaces $V$. However, as the coefficients are integers in $[n] = \{1,2,\ldots,n\}$, there are at 
most $n^m$ different inequalities, and so the linear maps $\mathbf{B}$ define a polytope $P_\mathbf{B}$ (sometimes called the {\em BL polytope}) in $\R^n$, such that 
$(\mathbf{B,p})$ is feasible iff $\mathbf{p} \in P_\mathbf{B}$. 
The BL polytopes $P_\mathbf{B}$ have received a lot of attention. For example, their vertices were characterized for the ``rank-1 case'' (namely, when all dimensions $n_j=1$) by Barthe \cite{Barthe2}, and this was extended to other cases by Valdimarsson~\cite{Vald10}. In the same paper, Valdimarsson considers the question of generating the inequalities defining $P_\mathbf{B}$, and gives a finite algorithm to do so (needless to say, after they are given, feasibility testing becomes a linear programming problem). No explicit upper bound  on the complexity of this algorithm is given, but it is at least exponential in $m$. The same holds for the algorithm in~\cite{christ2013communication}  for generating the inequalities defining $P_\mathbf{B}$. 

We give a polynomial time algorithm for the feasibility problem, and much more. Our algorithm also 
gives a ``separation oracle'' (namely finds a violated inequality when infeasible). 
Recall again that the exponents $\mathbf{p}$ in the BL datum $(\mathbf{B,p})$ are given by 
$p_j=c_j/d$ where $c_j, d$ are integers.

\begin{theorem}[Corollary \ref{cor:bl_feasibility} rephrased]\label{feasibility-alg}
There is an algorithm that on input $(\mathbf{B,p})$ of binary length $b$ and common denominator $d$ runs in 
time $\poly(b,d)$
and provides the following information:
\begin{description}[leftmargin= 1in]
	\item [      Membership oracle:] Tests if $\mathbf{p} \in P_\mathbf{B}.$
	\item [      Separation oracle:] When $\mathbf{p} \not\in P_\mathbf{B}$, provides a violated inequality, namely 
		a basis for a vector space $V$ in $\R^n$ such that $\dim (V) > \sum_j p_j \dim(B_jV).$
\end{description}
\end{theorem}

We believe that the ability to efficiently optimize over such a wide family of polytopes with 
exponentially many facets should be useful for different optimization problems via appropriate 
reductions to this setting. We will give some simple examples of such reductions in Section~\ref{sec:mat_int}.
One concrete challenge is e.g. trying to find such a reduction which 
embeds the Edmonds polytope~\cite{Edm65} of perfect matchings in a general (non-bipartite) 
graph into some BL polytope.

It is an interesting open problem if one can improve the algorithm to depend 
polynomially on $\log d$ instead of $d$. This
would allow for optimization over the BL polytopes in polynomial time via the ellipsoid algorithm. We believe even the current separation oracle for BL polytopes with polynomial dependence on $d$ should allow for approximate optimization over BL polytopes but we haven't been able to prove this yet. 

\paragraph{\textit{Characterizing the BL constant}}

The following theorem of Lieb~\cite{Lieb90} characterizes the BL constant $\BL(\mathbf{B,p})$ in all cases 
where the BL datum is feasible. 
The heart of the proof establishes that the optimal constant in any BL inequality is attained (or approached) by plugging in density functions of appropriate centered Gaussians. For such densities, the BL constant has a nice expression via the following identity:
$$
\int_{\mathbb{R}^n} \expon \left( - \pi x^T A x\right) dx = \Det(A)^{-1/2}
$$

\begin{theorem}[\cite{Lieb90}]\label{BLconstant}

Assume $(\mathbf{B,p})$ is feasible. Then 
$$\BL(\mathbf{B,p}) = 
\left[ \sup \frac{\prod_j (\det X_j)^{p_j}}{\det \left(\sum_j p_j B_j^{\dagger} X_j B_j \right)}  \right]^{1/2} $$
where the supremum is taken over all choices of positive definite matrices $X_j$ in dimension $n_j$, and 
$B_j^{\dagger}$ denotes the adjoint map corresponding to $B_j$. 
\end{theorem}

Thus the BL constant for given BL datum is a solution to an optimization problem. 
However, as defined, it is not even convex. We are not aware of any algorithms to compute the 
BL constant in general. Moreover, to the best of our knowledge no explicit bounds on the BL 
constant (when finite) in terms of the BL datum were known. We resolve both issues: provide an 
explicit bound and give a polynomial time algorithm for computing the BL constant to any accuracy.

\begin{theorem}[Corollary \ref{cor:bl_upperbound_rational} rephrased]\label{thmintro:BL_ub}
For any feasible \textnormal{BL} datum $(\mathbf{B,p})$ that has binary length $b$ and common denominator $d$, it holds that $\BL(\mathbf{B,p}) \leq \exp(O(b \log(bd)))$. 
\end{theorem}

\begin{theorem}[Theorem \ref{thm:computing_bl} rephrased]\label{constant-alg}
There is an algorithm that on input $(\mathbf{B,p})$ of binary length $b$ and common denominator $d$, and an accuracy parameter 
$\epsilon >0$, runs in time $\poly(b, d, 1/\epsilon)$ and computes a factor 
$(1+\epsilon)$-approximation of $\BL(\mathbf{B,p}).$ Furthermore, the algorithm outputs a scaling $\mathbf{B'}$ which is almost geometric i.e. $\BL(\mathbf{B'},\mathbf{p}) \le 1+\eps$.
\end{theorem}

\subsection{Efficient computation of the BL-constant via scaling}

The algorithm underlying the proof of Theorem \ref{constant-alg} will be shown (via a reduction) to be a special case 
of the operator scaling algorithm of Gurvits~\cite{gurvits2004}, which the current authors analyzed 
and proved to be polynomial time in~\cite{GGOW}. 

In this subsection, we explicitly describe the algorithm in this special case of BL-inequalities only, without referring to~\cite{GGOW}.
It will be instructive to see how
notion of scaling, which naturally exists in the theory of BL inequalities, can be used algorithmically. Moreover, it will become clear  how the algorithm implies (known and new) structural consequences to the BL theory.
This description will also help to motivate our reduction in the technical chapters which follow (which will provide the proof for its run-time). 

Let us return to the BL constant, and to an important family of BL data called {\em geometric}
\footnote{The analogous term in the operator scaling setting will be called {\em doubly stochastic}.} 
defined below. It was introduced by Ball~\cite{Ball} and extended by Barthe~\cite{Barthe2}.

\begin{definition}\label{def:normalization}
A BL datum $(\mathbf{B,p})$ is called {\em geometric}  if it satisfies the following conditions 
(with $I_k$ denoting the $k\times k$ identity matrix).
\begin{description}[leftmargin= 1in]
\item [		Projection:] For every $j\in [m]$, $B_j$ is a projection, namely $B_j B_j^{\dagger} = I_{n_j}$.
\item [		Isotropy:] $\sum_j p_j B_j^{\dagger} B_j = I_n$
\end{description}
\end{definition}

Ball proved that for the special case of geometric BL datum with $n_j=1$, the constant is always one. Barthe extended this to the general case\footnote{The analogous 
theorem for operator scaling is that the capacity of doubly stochastic operators is always $1$. Another family of inequalities where the 
constant is always $1$ are the ``discrete'' BL inequalities of \cite{CDKSY15}.}.
\begin{theorem}[\cite{Ball, Barthe2}]\label{geometric}
For every geometric BL datum $(\mathbf{B,p})$ we have $\BL(\mathbf{B,p})=1$.
\end{theorem}

A simple and natural action of the general linear groups on the Euclidean spaces involved (which simply performs a basis change in each space) yields an equivalence relation on BL data\footnote{This group action naturally calls for a study of BL from an invariant theory viewpoint, which indeed exists in a much more general context that will be relevant to us in several ways, namely that of {\em quiver representations}. An extensive survey is ~\cite{DW07}.}. Specifically, we say that $(\mathbf{B,p})$ and $(\mathbf{B',p'})$ are {\em equivalent} if there exist matrices $C\in GL_n(\R)$, $C_j \in GL_{n_j}(\R)$ (which are called {\em intertwining transformations}\footnote{They will be called {\em scaling matrices} in the operator scaling setting.} in~\cite{BCCT}) such that $B_j' = C_j^{-1}B_jC$ and $p_j' = p_j$ for all $j$. It is easy to compute the effect of such action on the BL constant.

\begin{proposition}[\cite{BCCT}]\label{basis-change}
Assume $(\mathbf{B,p})$ and $(\mathbf{B',p'})$ are equivalent via $C,C_j$ as above. Then
$$ \BL(\mathbf{B',p'}) = \frac{\prod_j (\det(C_j))^{p_j}}{\det (C)} \BL(\mathbf{B,p}).$$
where $\det$ is the determinant polynomial on matrices of the appropriate dimension.
\end{proposition}

Equivalence and this simple, efficiently computable formula suggests a natural path to computing BL constants. Given BL datum $(\mathbf{B,p})$, compute an equivalent {\em geometric} datum $(\mathbf{B',p'})$ (if one exists), and the intertwining transformations relating the two. The main questions are, is there such a geometric equivalent datum, and how to compute these transformations. The theory of (quantum) operator scaling~\cite{gurvits2004,GGOW} suggests that a simple, greedy procedure will work. This algorithm will underly (most of) the statements in Theorem~\ref{feasibility-alg} and Theorem~\ref{constant-alg}. As mentioned, we will explain the connection and reduction to operator scaling in the next sections, and here describe informally how it is applied to the BL setting.

\paragraph{\textit{The scaling algorithm}}

The key idea of {\em greedy, iterative} scaling, which goes back to Sinkhorn \cite{Sink} in the (classical) matrix scaling, is that when trying to achieve a pair of conditions as in the definition of geometric BL datum, we should try to satisfy them one at a time! Say we are given $(\mathbf{B,p})$ and assume for example that the isotropy condition is not met. A non-triviality condition (otherwise $(\mathbf{B,p})$ is not feasible) implies that $C \triangleq \sum_j p_j B_j^{\dagger} B_j $ is invertible, and so we can set $B_j \leftarrow B_j C^{-1/2}$ to get a new datum which satisfies isotropy. 
So, we ``only'' have to fix the projection property; lets do it. Again, non-triviality implies that each of 
$C_j \triangleq B_j B_j^{\dagger}$  
is invertible, and we can now set $B_j \leftarrow C_j^{-1/2} B_j $, satisfying the projection property. Of course, we may now have ruined isotropy. No problem: lets fix it again in the same way, and repeat alternately fixing the unsatisfied property. The magic is that this sequence of {\em normalization} steps converges, and moreover, converges in polynomial time, to a geometric datum, whenever the original datum $(\mathbf{B,p})$ is feasible!\footnote{Note that the sequence might only converge to a geometric datum and it is possible that no element of the sequence is geometric.} Let us describe the BL algorithm more precisely and then state its properties formally.

\begin{Algorithm}[H]
\textbf{Input}: A BL datum $(\mathbf{B,p})$ \\ \\
Set $\mathbf{B^0}=\mathbf{B}$.\\
Repeat for $i=1$ to $t$ normalization steps.
\begin{itemize}
\item if $i$ is odd, normalize $\mathbf{B^i}$ from $\mathbf{B^{i-1}}$ to satisfy isotropy.
\item if $i$ is even, normalize $\mathbf{B^i}$ from $\mathbf{B^{i-1}}$ to satisfy projection.
\end{itemize}
\textbf{Output}: $\mathbf{B^t}$. 
\caption{The BL scaling algorithm}
\label{The BL scaling algorithm}
\end{Algorithm}

The BL scaling algorithm defines a dynamics in the space of BL data which are restricted to stay in one equivalence class. If $(\mathbf{B,p})$ is the original datum which is not already geometric, it defines a sequence $\mathbf{B^0}=\mathbf{B}, \mathbf{B^1}, \mathbf{B^2},\dots$ which alternately satisfies the isotropy or projection properties (note that the exponent vector $p$ is used, but not changed in this process). We now state the main three properties which underlie the analysis. 

\begin{theorem}\label{analysis}
For every $(\mathbf{B,p})$ the following holds:
\begin{enumerate}
\item \emph{Upper bound:} If $(\mathbf{B,p})$ is feasible, then $\BL(\mathbf{B,p}) \leq \exp(O(b \log(bd)))$.
\item \emph{Lower bound:} If the datum $(\mathbf{B,p})$ is either isotropy-normalized or projection-normalized, then $\BL(\mathbf{B,p}) \geq 1$. 
\item \emph{Progress per step:} Let $(\mathbf{B',p})$ denote the result of applying either isotropy-normalization or projection-normalization to $(\mathbf{B,p})$. Then $($as long as $\BL(\mathbf{B,p}) \geq 1+\epsilon)$, we have  $$\BL(\mathbf{B',p}) \leq (1-\poly \left( \epsilon/nd \right)) \BL(\mathbf{B,p}).$$
\end{enumerate}
\end{theorem}

We will not need to prove this theorem, as it will follow via our reduction\footnote{In the case when the exponent vector $\mathbf{p}$ is irrational, the reduction to operators does not exist and so in this case, we provide an analysis of the number of iterations of Algorithm \ref{The BL scaling algorithm} in Section \ref{sec:irrational}.} from the analogous theorem in the more general operator scaling setting\footnote{Where statements made here regarding the BL constants will be replaced by statements about a related notion called the {\em capacity} of an operator.}. However, let us say a few words here about what goes into each item above using BL language.  Property (3) on progress per step is the simplest: it follows from a robust version of the AM-GM inequality. Property (2) follows from plugging in densities of spherical Gaussians (it was proved in \cite{Vald11}, also see Section \ref{sec:properties}).
 Property (1) is the hardest, in that it requires degree bounds on the generating polynomial invariants in the null-cone of the quiver associated to the group action above. We will briefly explain this mouthful in Section \ref{sec:ub_bl}, and refer the reader to~\cite{GGOW} for details. In Section \ref{sec:ub_bl}, we will also give an improved and simplified analysis of the main theorem from \cite{GGOW} and derive the desired upper bounds on the BL constant via our main reduction. All in all, it is worth stressing again that beyond this non-trivial use of algebra, on the analytic side the most general BL inequalities follow from the most basic one, AM-GM.

This analysis above implies that we can choose the number $t$ of steps in the BL scaling algorithm 
to be $\poly(b,d, 1/\eps)$ so that $\BL(\mathbf{B^{i},p}) \leq 1+\eps$ for some $1 \le i \le t$ 
(whenever the initial datum $(\mathbf{B^0,p})$ is feasible). Testing if the lower bound is met
($\BL(\mathbf{B^{i},p}) \leq 1+\eps$) decides feasibility\footnote{It is not immediately clear how to 
test $\BL(\mathbf{B^{i},p}) \leq 1+\eps$. However this can be done by checking closeness to the 
geometric position in an appropriate metric, similar to the distance to double stochasticity for 
operators that will be defined later.}, thus proving the first item in Theorem~\ref{feasibility-alg}. 
For proving the second item
we will employ a very different, beautiful combinatorial algorithm for operator scaling (only the decision version of it), that was 
developed after ours by Ivanyos, Qiao and  Subrahmanyam~\cite{IQS15b}. Their algorithm relies 
on even tighter degree bounds from invariant theory which were proved very recently~\cite{DM15}, 
and has several advantages over the algorithm in~\cite{GGOW}; in particular it can find violated 
dimension inequalities. Going back to Theorem~\ref{constant-alg} on approximating the BL 
constant, note that it suffices to multiply together the (respective) normalizing matrices used in all 
steps of the BL scaling algorithm, using the formula in Proposition~\ref{basis-change}.

We now return to the scaling algorithm above, and discuss connections and applications of its structure and properties  in the next two bullets.

\paragraph{\textit{The BL constant, scaling and alternate minimization}}

The problem of computing the BL constant, as it is formulated here, is a non-convex problem.
We presented an efficient algorithm for this optimization problem which may be viewed as an instance of a general heuristic called {\em alternate minimization}. Let us describe and discuss this method.

In general, alternate minimization
is a very general technique introduced by~\cite{CT84}, which was devised to solve problems of the following 
kind. Given some universe $U$, a (distance) function $d : U \times U \to \R_+$ and two subsets 
$X, Y \subseteq U$, find 
$\dst (\bx^*, \by^*) = \textnormal{arg}\min_{\substack{\bx \in X, \by \in Y}} d(\bx, \by)$. We note that 
such problems are in general NP-hard, even over convex sets $X,Y$ and natural distances $d$. 
Nevertheless, alternate minimization is widely used to solve special cases of such problems in practice.\footnote{Note that such problems can be defined with more than two components, and the alternate minimization approach below can be extended to such problems, but we stick here with two.}

An alternate minimization algorithm for the problem above starts with a arbitrary point $(x_0,y_0) \in X \times Y$ and generates a sequence of points $(x_i,y_i)$ as follows. In even steps $i$, $x_{i+1}=x_i$, and $y_{i+1}$ is chosen to be the value of $y$ which minimizes $d(x_i,y)$ over the second coordinate alone (an optimization problem that is often far simpler). In odd steps the roles of $x$ and $y$ are reversed. It is easy to cast the scaling algorithm above as a process of this nature. It is also easy to see that it enjoys some nice properties in general situations. For example, if both $X$ and $Y$ are strictly convex, and $d$ is a metric, then this process converges to the optimum. Of course, the speed of convergence is the main question. We list below some other, possibly familiar instances of alternate minimization for which such convergence is not known.

The famous Lemke-Howson algorithm for finding a Nash equilibrium for 2-player 
games~\cite{Lemke-Howson} is of that nature. Recall that in that algorithm one starts with an arbitrary strategy for both players, and then one alternates fixing the strategy of one, and finding the best response for the other via linear programming.
 Another such example is the work of Zafeiriou and Petrou~\cite{ZP10}
on computing non-negative tensor factorization. As above, the analysis of such alternating 
minimization algorithms just provide convergence of the procedure, without proving {\em how fast}
it converges to the minimum, with the hope that this heuristic will lead to quick convergence on instances arising in practice. 
Of course, there are examples,  e.g. the alternate minimization algorithms for matrix completion in~\cite{JNS13, hardt14} which prove rapid convergence under certain conditions on the inputs.

In this work, as well as in~\cite{GGOW}, we prove that our alternate minimization algorithm 
not only converges, but also converges to the infimum in polynomial time. Our proof is based on the introduction
of a potential function which measures how much progress we make in every step of the alternating
minimization algorithm. In our case the potential function {\em is} the very BL constant we are trying to optimize, which is a special case of the so-called {\em capacity} we use as potential  in our operator scaling algorithm~\cite{GGOW}. As it happens, the analysis of convergence relies on a combination of algebraic tools, in particular from  (the commutative) invariant theory and (the non-commutative) theory of skew fields, as explained in that paper. We hope that such potential functions as well as the methods to analyze their progress, which we call  {\em capacity methods}, will be of use in the analyses of other problems which use alternate 
minimization, especially when the result of optimizing each of the solution components in individual optimization steps may be viewed as the action of a group. We are currently exploring this direction for alternate minimization with more than two components, and these lead to interesting algebraic questions related to questions that arise in invariant theory as well as in Geometric Complexity Theory of Mulmuley and Sohoni~\cite{GCT1, GCT5}.

\paragraph{\textit{Bounds and continuity of the BL constant, and non-linear BL inequalities}}

Another important property of the BL scaling algorithm above, evident from its form, is that it is smooth!  Namely, the orbits under the algorithm of two sufficiently close BL data will remain close to each other for the duration of the algorithm. This implies that the BL constants of both will be close as well, and moreover this can be quantified! This simple observation gives a much stronger result than existing qualitative results on the continuity of BL constants which has recently received significant attention. This is yet another example   highlighting the usefulness of algorithmic results in mathematics. We briefly overview the motivation and known results before stating ours.

How smooth (or regular) the BL constant is (in terms of the BL data) is of course a natural question, especially as the expression in Lieb's theorem~\ref{BLconstant} is not even convex. Further motivation to study this question arises from a variety of non-linear variants of the BL inequalities, in which the maps $B_j$ are non-linear, but at least smooth enough so as to be approximable by linear ones in a small ball.
Some such variants, needed in specific applications, were proved directly.  The papers~\cite{BB10,BBFL15} ask this question in full generality, and demonstrate the importance of such smoothness conditions of the BL constant for general non-linear extensions of BL inequalities. These papers also
exposit the many diverse applications of such generalized BL inequalities
in analysis, number theory and other areas (possibly the most impressive recent ones are the applications of ``Kakeya-type''  BL inequalities to number theory in~\cite{BoDe15,BDG15}, resolving long standing open problems including the ``Vinogradov mean-value conjecture''). 

Local boundedness, a condition weaker than continuity, is established in~\cite{BBFL15}.

\begin{theorem}[\cite{BBFL15}]\label{locally-bounded}
If $(\mathbf{B,p})$ is feasible, then there are $\delta > 0$ and $C < \infty$ such that for all $(\mathbf{B',p})$ 
such that $\|\mathbf{B}'-\mathbf{B}\|_2 \leq \delta$ we have $\BL(\mathbf{B',p}) \leq C$.
\end{theorem}

This was followed recently by the paper~\cite{BBCF16}, which proves continuity of the BL constant.

\begin{theorem}[\cite{BBCF16}]\label{continuous}
The \textnormal{BL} constant $\BL(\mathbf{B,p})$ is continuous in $\mathbf{B}$. Namely, for every feasible $(\mathbf{B,p})$ and $\epsilon>0$, there exists a $\delta >0$ such that
if $\|\mathbf{B}'-\mathbf{B}\|_2 \leq \delta$, then we have that $|\BL(\mathbf{B',p})-\BL(\mathbf{B,p})| \leq \epsilon$.
\end{theorem}

In both theorems above, no quantitative bounds are given. This is not surprising, as the proofs of both use compactness.
On the other hand, the availability of a smooth algorithm, together with the quantitative analysis of its continuity provided in~\cite{GGOW}, imply via our reduction the following quantitative bound on the continuity\footnote{While Theorem \ref{stability} only implies continuity of $\BL(\mathbf{B,p})$ for rational $\mathbf{B}$ and $\mathbf{p}$, we describe in Section \ref{sec:cont_bl} how our proof can be extended to the case of irrational $\mathbf{B}$ and speculate how this might also extend to irrational $\mathbf{p}$.} of the BL constant (giving a multiplicative approximation).

\begin{theorem}[Theorem \ref{BL_continuity} rephrased]\label{stability}
For every parameters $b, d, \eps >0$ there exists a 
$$
\delta = \delta(b,d,\eps) \leq \expon \left( -\frac{\poly(b,d)}{\eps^3} \right)
$$ 
such that for every \textnormal{BL} data $(\mathbf{B,p})$ of size $b,d$ and every \textnormal{BL} data $(\mathbf{B',p})$ such that $\|\mathbf{B}'-\mathbf{B}\|_2 \leq \delta$, we have that if $\BL(\mathbf{B,p})$ is finite, then $\BL(\mathbf{B',p})$ is finite. Furthermore, 
$$
(1-\eps) \BL(\mathbf{B,p}) \leq  \BL(\mathbf{B',p}) \leq (1+\eps)\BL(\mathbf{B,p})
$$
\end{theorem}

We hope that  this explicit bound will be useful to making non-linear BL inequalities quantitative, and to their applications.

\subsection{Applications of Brascamp-Lieb in Computer Science, Optimization and Beyond}

Besides the applications to other areas of mathematics as we mentioned above,  the structural theory of
BL inequalities (sometimes with the connection to operator scaling) have encountered several applications in computer science and optimization. We give below a (partial) survey of these applications. 

Many diverse applications exist  even for the special case of rank-$1$ BL setting (when all the $n_i$'s are $1$), which has been studied extensively. In this case, the BL polytopes are exactly the {\em basis polytopes} \cite{Barthe2}. More specifically, if the maps $B_i$ are given by $B_i(x) = \langle v_i,x \rangle$, then the BL polytope is the convex hull of the following set: 
$$
\{1_I\,:\, I \subset [m], |I|=n, (v_i)_{i \in I} \: \text{forms a basis for $\mathbb{R}^n$}\}
$$ 
Hence Theorem \ref{feasibility-alg} for this special case is well known. Furthermore, computing the BL constant in this special case can be formulated as a convex problem and hence Theorem \ref{constant-alg} in this special case can be obtained by the ellipsoid algorithm \cite{gur-sam1, HardtM13} (\cite{gur-sam1} in fact solves a more general problem).  We list some of the applications of this setting.
\begin{itemize}
\item In \cite{HardtM13}, Hardt and Moitra provide an application of scaling the rank-$1$ BL datum to a geometric position (which they call isotropic) to the problem of {\em robust subspace recovery}. 
\item Forster \cite{Forster02} used the existence of a scaling to geometric position for rank-$1$ BL datum (sometimes called Barthe's theorem) to give the first nontrivial lower bounds for the {\em sign rank} of an explicit matrix (in particular the Hadamard matrix), and thereby prove the first nontrivial {\em unbounded error communication complexity} lower bounds for the inner product function. 
\item In~\cite{DSW_LCC14}, Dvir et al. use Barthe's theorem as part of the first super-quadratic lower bounds on the size of 3-query {\em LCCs  (Locally Correctable Codes)} over the Real numbers. 
\item Recently, Nikolov and Singh \cite{NikolovS16} use a optimization problem similar to the rank-$1$ BL constant to give a polynomial time algorithm for approximating the maximum value of the determinant of a submatrix of a psd matrix under partition constraints. 
\end{itemize}

The class of BL polytopes is a rich class of polytopes. As mentioned above, the rank-$1$ case gives rise to basis polytopes. The rank-$2$ case (when all $n_i$'s are $1$ or $2$) has been completely characterized by Valdimarsson \cite{Vald10}. An important special case is given by the {\em linear matroid intersection} polytopes which we describe next. Suppose $(v_1,\ldots,v_m)$ and $(w_1,\ldots,w_m)$ are two collection of vectors in $\mathbb{R}^n$. Then the matroid intersection polytope given by these collections is the convex hull of the following set:
$$
\{1_I \,:\, I \subset [m], |I|=n, (v_i)_{i \in I}, (w_i)_{i \in I} \: \text{both form a basis for $\mathbb{R}^n$}\}
$$
Let us define the maps $B_i : \mathbb{R}^{2n} \rightarrow \mathbb{R}^2$ in the following way: 
$$
B_i(x,y) = (\langle v_i,x\rangle, \langle w_i, y\rangle)
$$
Then the BL polytope $P_{\mathbf{B}}$ is exactly the corresponding matroid intersection polytope. In particular, the bipartite matching polytopes are a special case of rank-$2$ BL polytopes. 

We will give proofs of these statements in Section~\ref{sec:mat_int}, with the hope of demonstrating the simplicity by which such exponential size linear programs of this form capture (simple) combinatorial optimization problems.
As mentioned before, it remains a very interesting challenge to see if the general matching polytopes (or in general the linear matroid matching polytopes) are a special case of BL polytopes, for which it may still suffice to stay in the rank-2 case. We  note a peculiar discrepancy between operator scaling and BL-inequalities when encoding the matroid intersection problem. In operator scaling, this problem was one of the first special cases considered in Gurvits' paper~\cite{gurvits2004}, and he showed that one could encode it by an operator scaling problem with matrices of rank $1$. Here, when encoded as a BL polytope, we use rank-$2$ projectors, and believe it is impossible to do in rank $1$. Even though BL is a special case of operator scaling (as our main reduction shows) the relative power of these two formulations with respect to rank is not clear.

Characterizing the vertices of BL-polytopes for projectors of rank $3$ and higher remains a very interesting question, as is the case of rank-$2$ matrices in the formulation of operator scaling.

BL-inequalities (and operator scaling) have recently found applications in computational and 
combinatorial geometry in the paper of  Dvir et al.~\cite{DGOS16}. They  generalize 
the rank lower bounds for design matrices given in~\cite{BDWY12, DSW14} to the setting of 
block matrices. This result allows them to obtain sharp bounds on Sylvester-Gallai type 
theorems for arrangements of subspaces, to obtain new incidence bounds for high-dimensional
line/curve arrangements, as well as to prove structural rigidity results in the projective setting.

Beyond {\em applying}  BL-inequalities, we note that in recent years computer scientists started getting interested 
in {\em proving} them efficiently.  Efficiency here here should be taken to mean in the Sum-of-Squares (SoS) 
framework, a hierarchy of semi-definite programs that is perhaps the most powerful general algorithmic technique 
for a variety of optimization problems. The paper~\cite{BKS14}, generalizing from many particular examples on 
the power of SoS algorithms, suggests a general framework for rounding solutions in this system, in which a 
crucial element is efficient SoS proofs of inequalities (which in all known applications are special cases of 
BL-inequalities). This viewpoint has lead to many new recent SoS algorithms~\cite{BKS15, MSS16, HSSS16}, 
and naturally also lead to the question of efficiently proving (in this sense) all BL-inequalities. This was taken up by 
Lei and Sheng~\cite{LS16}, who have (in a very general setting)  SoS proofs of degree 
$\poly (d)$ for BL-inequalities that have denominator $d$ in its exponent $\mathbf{p}$-vector. This yields an algorithm for 
testing feasibility that is {\em exponential} in $d$. Our algorithms depend polynomially in $d$, and raise the 
question of relating operator scaling and SoS algorithms, and whether such a relation can help improve the 
bounds above. Needless to say, relating the power of these two optimization methods is motivated for many 
other reasons!

\paragraph{\textit{Dictionary}}

We give here a basic dictionary describing translation between central notions in the theory of Brascamp-Lieb inequalities and completely positive operators. This will be useful in reading the technical sections below.
The basic notions concerning completely positive operators will be defined in Sections \ref{sec:square} and \ref{sec:rect}. 


\begin{center}
    \begin{tabular}{| l | l |}
    \hline
    \textbf{Brascamp-Lieb inequalities} & \textbf{Completely positive operators} \\ \hline
    BL datum & Kraus operators of completely positive operators \\ \hline
    Intertwining transformations & Scaling operations \\ \hline
    Geometric BL datum & Doubly stochastic operator \\ \hline
    BL constant & Capacity \\ \hline
    \end{tabular}
\end{center}

\paragraph{\textit{Organization of the paper}}
All but the last two sections are devoted to providing background and then proving the main technical results of this paper.
Preliminaries and results from~\cite{gurvits2004,GGOW} about (square) completely positive operators, their scaling, capacity and more will be reviewed in Section \ref{sec:square}. In Section \ref{sec:rect}, we will extend these results to {\em rectangular} completely positive operators via a simple reduction. In Section \ref{sec:bl-to-cap}, we describe our main reduction, from BL data to (rectangular) completely positive operators, and prove Theorem \ref{feasibility-alg}. We also show how this reduction yields, as a corollary, a simple proof of the characterization (Theorem \ref{feasibility}) of \cite{BCCT}. In Section \ref{sec:ub_bl}, we describe an upper bound on the Brascamp-Lieb constants and prove Theorem \ref{thmintro:BL_ub}. In Section \ref{sec:compute_bl}, we show how to approximate Brascamp-Lieb constants by using the main reduction above and our operator scaling algorithm (proving Theorem \ref{constant-alg}). In Section \ref{sec:cont_bl}, we use the smoothness of the scaling algorithm to give explicit bounds for the continuity of Brascamp-Lieb constants and prove Theorem \ref{stability}. In Section~\ref{sec:mat_int} (which is independent from all others) we exemplify the power of exponentially large linear programs arising as BL-polytopes to capture certain simple combinatorial optimization problems. Section~\ref{sec:properties} contains some additional properties of the Brascamp-Lieb constant, strengthening the results in \cite{Vald11}. Section~\ref{sec:irrational} contains an analysis of the number of iterations required in Algorithm \ref{The BL scaling algorithm} to get close to geometric position. Finally, we conclude with some open problems in Section \ref{sec:open}.




\section{Square completely positive operators}\label{sec:square}

In this section, we introduce (square) completely positive operators and important facts and theorems about them which we will need \cite{gurvits2004, GGOW}. 

\noindent Given a complex matrix $A$, we will use $A^{\dagger}$ to denote the conjugate-transpose of $A$. For matrices 
with real entries this will just be $A^{T}$. $\succeq$ will be used to denote the Loewner order defined by positive semidefinite matrices. So we will write $A \succeq B$ if $A-B$ is positive semidefinite and $A \succ B$ if $A - B$ is positive definite.

\begin{definition}[\textbf{Completely Positive Operators}] An operator (or map) $T: M_n(\mathbb{C}) \to M_n(\mathbb{C})$ is called 
completely positive if there are $n \times n$ complex matrices 
$A_1, \ldots, A_m$ s.t. $T(X) = \sum_{i=1}^m A_i X A_i^{\dagger}$. The matrices $A_1,\ldots, A_m$ 
are called Kraus operators of $T$ (they need not be unique). $T$ is called completely positive trace 
preserving (cptp) if in addition, $\tr(T(X)) = \tr(X)$ for all $X$. This is equivalent to the condition $\sum_{i=1}^m A_i^{\dagger} A_i = I$.
\end{definition}

\begin{remark}
The above is actually not the usual definition of completely positive operators.  $T$ is defined to be positive if $T(X) \succeq 0$ whenever $X \succeq 0$. $T$ is completely positive if $I_n \tensor T$ is positive for all $n \ge 1$. Choi \cite{Choi} proved that an operator is completely positive iff 
it is of the form stated above.
\end{remark}

\begin{definition}[\textbf{Tensor products of operators}]\label{def:tensor} Given operators $T_1: M_{d_1}(\mathbb{C}) \to M_{d_1}(\mathbb{C})$ and $T_2: M_{d_2}(\mathbb{C}) \to M_{d_2}(\mathbb{C})$, we define their tensor product 
$T_1 \tensor T_2 : M_{d_1d_2}(\mathbb{C}) \to M_{d_1d_2}(\mathbb{C})$ in the natural way
$$
(T_1 \tensor T_2) (X \tensor Y) = T_1(X) \tensor T_2(Y)
$$
and extend by linearity to the whole of $M_{d_1d_2}(\mathbb{C})$. 
\end{definition}

\begin{definition}\label{dual,ds}
If $T(X) = \sum_{i=1}^m A_i X A_i^{\dagger}$ is a completely positive operator, we define its dual $T^*$ by
$T^*(X) = \sum_{i=1}^m A_i^{\dagger} X A_i$. If both $T$ and $T^*$ are trace preserving, namely $T(I) = T^*(I) = I$
then we call $T$ (and $T^*$) {\em doubly stochastic}.
\end{definition}

\begin{definition}[\textbf{Rank Decreasing Operators}, \cite{gurvits2004}] A completely positive operator $T$ is said to be rank-decreasing 
if there exists an $X \succeq 0$ s.t. 
$\text{rank}(T(X)) < \text{rank}(X)$.
\end{definition}

\noindent Now that we defined completely positive operators, we define their {\em capacity}, which is a very important complexity measure
of such operators suggested in~\cite{gurvits2004}. For a special subset of operators, their capacity will be related to the Brascamp-Lieb constants.

\begin{definition}[\textbf{Capacity}, \cite{gurvits2004}] The capacity of a completely positive operator $T$, denoted by $\capac(T)$, is defined as
$$\capac(T) = \text{inf} \{ \text{Det}(T(X)) : \text{$X \succ 0$, Det$(X) = 1$}\}$$
\end{definition}

\begin{remark} In the definition of capacity, we didn't specify whether to optimize over complex psd matrices or real psd matrices. However, as long as $T$ is defined by 
real Kraus operators, the value of capacity doesn't change if we optimize over complex or real psd matrices. This can be seen by the fact that an almost optimizing solution to capacity
can be obtained by a natural iterative sequence (Theorem \ref{thm:cap_scaling}).
\end{remark}

\noindent Next we define the notion of operator scaling.

\begin{definition}[\textbf{Operator Scaling}, \cite{gurvits2004}] An operator $T'$ is called an operator scaling of $T$ if there exist invertible matrices $B,C$ s.t. 
$$
T'(X) = B \cdot T\left( C \cdot X \cdot C^{\dagger}\right) \cdot B^{\dagger}
$$
Alternatively, if $A_1,\ldots,A_m$ are Kraus operators for $T$, then $B A_1 C,\ldots,BA_mC$ are Kraus operators for $T'$. 
\end{definition}

\noindent Next we define a ``distance" of an operator from being doubly stochastic. 

\begin{definition}[\textbf{Distance to double stochasticity}]
$$
\ds(T) = \tr \left[ \left(T(I) - I\right)^2 \right] + \tr \left[ \left( T^*(I)   - I\right)^2 \right]
$$
\end{definition}

\noindent Next we define two operations on operators, which enforce one of the conditions of being doubly stochastic.

\begin{definition}[\textbf{Right Normalization}] Given a completely positive operator $T$, define its {\em right normalization} $T_R$ as follows:
\begin{align*}
T_R(X) = T \left(T^*(I)^{-1/2} \cdot X \cdot  T^*(I)^{-1/2} \right)
\end{align*}

\end{definition}

\noindent Note that $T_R^*(I) = I$. This is because $T_R^*(X) = T^*(I)^{-1/2} \cdot T^*(X) \cdot T^*(I)^{-1/2}$.

\begin{definition}[\textbf{Left Normalization}] Given a completely positive operator $T$, define its {\em left normalization} $T_L$ as follows:
\begin{align*}
T_L(X) =  T(I)^{-1/2} \cdot T(X) \cdot T(I)^{-1/2}
\end{align*}

\end{definition}

\noindent Note that $T_L(I) = I$. 
\\
\\
\noindent Algorithm \ref{Gurvits_alg} was suggested by Gurvits \cite{gurvits2004} to find a doubly stochastic scaling of an operator $T$. 

\begin{Algorithm}
\textbf{Input}: Completely positive operator $T$. 

Perform right and left normalizations starting from $T_0 = T$ alternately for $t$  steps. Let $T_j$ be the operator after $j$ steps. 

\textbf{Output}: $T_t$

\caption{Algorithm $G$}
\label{Gurvits_alg}
\end{Algorithm}

\begin{theorem}[\cite{gurvits2004}]\label{thm:cap_RND} Let $T$ be a completely positive operator. Then $T$ is rank non-decreasing iff $\capac(T) > 0$. 
\end{theorem}

\begin{theorem}[\cite{gurvits2004}]\label{thm:cap_scaling} Let $T$ be a completely positive operator s.t. $\capac(T) > 0$. Then for every $\eps > 0$, there exists an operator scaling $T'$ of $T$ s.t. $\ds(T') \le \eps$. Furthermore this $T'$ can be found by running Algorithm $G$ for an appropriate number of steps $t(T, \eps)$. 
\end{theorem}

\noindent The first polynomial time algorithms for checking whether $T$ is rank non-decreasing and also computing a multiplicative approximation to capacity were given by \cite{GGOW}. In the algorithms below, the input will be a completely positive operator $T$ which is given in terms of its Kraus operators $A_1,\ldots, A_m$. We will assume that the entries of the matrices $A_i$ can be described using $b$ bits.

\begin{theorem}[\cite{GGOW}] Let $T: M_n(\mathbb{C}) \rightarrow M_n(\mathbb{C})$ be a completely positive operator s.t. the bit-complexity of the description of $T$ $($in terms of its Kraus operators$)$ is $b$. Then there is a $\poly(n,b)$ time $($deterministic$)$ algorithm to test whether $T$ is rank non-decreasing. There is also a $\poly(n,b,1/\eps)$ time algorithm to compute a $(1+\eps)$-multiplicative approximation to $\capac(T)$. Furthermore, the algorithm outputs a scaling $T'$ of $T$ which is almost doubly stochastic i.e. $\capac(T') \ge 1-\eps$. 
\end{theorem}

\noindent The above mentioned algorithm does not find a witness to the rank decreasing property. Ivanyos, Qiao and 
Subrahmanyam \cite{IQS15b} were the first ones to give an efficient algorithm that computes such a witness. Their 
algorithm, however, is quite different from the one in \cite{GGOW} and does not compute capacity.

\begin{theorem}[\cite{IQS15b}]\label{alg:IQS} Let $T: M_n(\mathbb{C}) \rightarrow M_n(\mathbb{C})$ be a completely positive operator s.t. the 
bit-complexity of the description of $T$ $($in terms of its Kraus operators$)$ is $b$. Then there is a $\poly(n,b)$ time 
$($deterministic$)$ algorithm to test whether $T$ is rank non-decreasing. Furthermore, if $T$ is rank decreasing, then the 
algorithm also outputs an $X \succeq 0$ s.t. $\text{Rank}(T(X)) < \text{Rank}(X)$.
\end{theorem}


\section{Rectangular operators}\label{sec:rect}

In this section we extend the definitions and main results of the theory of square operators,
presented in~\cite{gurvits2004, GGOW}, to rectangular operators, which we will define below. In particular, at the end of the section (Corollaries \ref{cor:GGOW} and \ref{cor:IQS_rect}), we make explicit what the operator scaling algorithms for square operators imply for rectangular operators.


We define a rectangular operator as any linear map $T : \cM_{n_1}(\C) \to \cM_{n_2}(\C)$. 
Choi's characterization~\cite{Choi} extends to rectangular operators as well. And this is
how we define completely positive rectangular operators.

\begin{definition}[\textbf{Rectangular Completely Positive Operators}] An operator (or map) $T: M_{n_1}(\mathbb{C}) \to M_{n_2}(\mathbb{C})$ is called 
completely positive if there are $n_2 \times n_1$ complex matrices 
$A_1, \ldots, A_m$ s.t. $T(X) = \sum_{i=1}^m A_i X A_i^{\dagger}$. The matrices $A_1,\ldots, A_m$ 
are called Kraus operators of $T$ (and they are not unique). $T$ is called completely positive trace 
preserving (cptp) if in addition, $\tr(T(X)) = \tr(X)$ for all $X$. This is equivalent to the condition $T^*(I) = \sum_{i=1}^m A_i^{\dagger} A_i = I$.
\end{definition}

The following is a natural extension of the definition of doubly stochastic operators to the rectangular case.

\begin{definition}\label{dual,ds}
If $T(X) = \sum_{i=1}^m A_i X A_i^{\dagger}$ is a completely positive operator, we define its dual $T^*$ by
$T^*(X) = \sum_{i=1}^m A_i^{\dagger} X A_i$. We say $T$ is doubly stochastic if 
$T \left(\frac{n_2}{n_1} I_{n_1} \right) = I_{n_2}$ and $T$ is trace preserving. 
This is the same as 
$$T \left(\frac{n_2}{n_1} I_{n_1} \right) = I_{n_2} \ \text{ and} \ 
T^*(I_{n_2}) = I_{n_1}. $$ 
\end{definition}

\begin{remark} Note that there is an asymmetry in the above definition w.r.t to $T$ and $T^*$. This is an arbitrary convention on our part just to ensure that doubly stochastic operators are trace preserving. Not much in the theory would change if one defined doubly stochastic in a different way.
\end{remark}

The following is a natural extension of the rank decreasing property to the rectangular case. For a matrix $X$ of dimension $n$ we defined its {\em fractional-rank} as $\text{Rank}(X)/n$. 

\begin{definition}[\textbf{Fractional-Rank Decreasing Operators}] A completely positive operator $T$ is said to be fractional-rank decreasing 
if there exists an $X \succeq 0$ s.t. 
$\frac{\text{rank}(T(X))}{n_2} < \frac{\text{rank}(X)}{n_1}$.
\end{definition}

The following is a natural extension of capacity to the rectangular case.

\begin{definition}[\textbf{Capacity}] The capacity of a completely positive operator $T: M_{n_1}(\mathbb{C}) \to M_{n_2}(\mathbb{C})$, denoted by $\capac(T)$, is defined as
$$
\capac(T) = 
\inf \left\{ \frac{\Det \left( \frac{n_2}{n_1} T(X) \right)}{\Det(X)^{\frac{n_2}{n_1}}} : X \succ 0 \right\} 
= \inf \left\{ \Det \left(\frac{n_2}{n_1} T(X) \right) : X \succ 0, \Det(X) = 1\right\} 
$$
\end{definition}

\begin{definition}[\textbf{Operator Scaling}] An operator $T'$ is called an operator scaling of $T$ if there exist invertible matrices $B,C$ (which are $n_2 \times n_2$ and $n_1 \times n_1$ respectively) s.t. 
$$
T'(X) = B \cdot T\left( C \cdot X \cdot C^{\dagger}\right) \cdot B^{\dagger}
$$
Alternatively, if $A_1,\ldots,A_m$ are Kraus operators for $T$, then $B A_1 C,\ldots,BA_mC$ are Kraus operators for $T'$. 
\end{definition}

\begin{definition}[\textbf{Distance to double stochasticity}]
$$
\ds(T) = \tr \left[ \left( T \left(\frac{n_2}{n_1} I_{n_1} \right) - I_{n_2} \right)^2\right] + 
\tr \left[ \left( T^*(I_{n_2}) - I_{n_1} \right)^2\right]
$$
\end{definition}

The next two operations are natural extensions of the right and left normalizations to the case of rectangular operators. 

\begin{definition}[\textbf{Right Normalization}] Given a completely positive operator $T: M_{n_1}(\mathbb{C}) \to M_{n_2}(\mathbb{C})$, define its {\em right normalization} $T_R$ as follows:
\begin{align*}
T_R(X) = T \left(T^*(I_{n_2})^{-1/2} \cdot X \cdot  T^*(I_{n_2})^{-1/2} \right)
\end{align*}

\end{definition}

\noindent Note that $T_R^*(I_{n_2}) = I_{n_1}$. 

\begin{definition}[\textbf{Left Normalization}] Given a completely positive operator $T: M_{n_1}(\mathbb{C}) \to M_{n_2}(\mathbb{C})$, define its {\em left normalization} $T_L$ as follows:
\begin{align*}
T_L(X) =  \frac{n_1}{n_2} \cdot T \left(I_{n_1} \right)^{-1/2} \cdot T(X) \cdot T \left(I_{n_1}\right)^{-1/2}
\end{align*}

\end{definition}

\noindent Note that $T_L \left(\frac{n_2}{n_1} I_{n_1} \right) =  I_{n_2}$. 
\\
\\
\noindent Algorithm \ref{Gurvits_alg_rect} is a natural extension of Algorithm $G$ to the case of rectangular operators.

\begin{Algorithm}
\textbf{Input}: Completely positive operator $T : M_{n_1}(\mathbb{C}) \rightarrow M_{n_2}(\mathbb{C})$. 

Perform right and left normalizations starting from $T_0 = T$ alternately for $t$  steps. Let $T_j$ be the operator after $j$ steps. 

\textbf{Output}: $T_t$

\caption{Algorithm $G$ for rectangular operators}
\label{Gurvits_alg_rect}
\end{Algorithm}

We will now see how from a rectangular operator $T: M_{n_1} (\mathbb{C}) \to M_{n_2}(\mathbb{C})$, 
we can define a square operator $\widetilde{T} : M_{n_1 n_2}(\mathbb{C}) \to M_{n_1 n_2}(\mathbb{C})$ 
that captures most of the properties of $T$ (we can also embed in dimension $\text{lcm}(n_1,n_2)$ but 
since it only saves a quadratic factor, we ignore this to simplify notation). $\widetilde{T}$ is intended to act on 
block diagonal matrices (the off-diagonal blocks are ignored).  

\begin{construction}\label{con:squaring-operator}
	Let $T: M_{n_1} (\mathbb{C}) \to M_{n_2}(\mathbb{C})$ be an operator. 
	Given an $n_1 n_2 \times n_1 n_2$ psd matrix $X$, view it as a $n_2 \times n_2$ block matrix, where 
	each block $X_{i,j}$ is $n_1 \times n_1$. Define $\widetilde{T}(X)$ as $n_1 \times n_1$ block diagonal 
	matrix where each diagonal block is the same $n_2 \times n_2$ matrix: 
	$\frac{1}{n_1} \sum_{i=1}^{n_2} T(X_{i,i})$. In other words, we have that
	$$ \widetilde{T}(X) = I_{n_1} \otimes \left(\frac{1}{n_1} \sum_{i=1}^{n_2} T(X_{i,i}) \right)  $$
	In terms of Kraus operators, if $A_1,\ldots,A_m$ are Kraus operators for $T$, then a set of Kraus operators
	for $\widetilde{T}$ are given by the set $\left\{\frac{1}{n_1} E_{i,j} \tensor A_k\right\}_{i=1,j=1,k=1}^{n_1, n_2, m}$. 
	Here $E_{i,j}$ is the elementary matrix with a $1$ in position $(i,j)$ and $0$'s everywhere else. 
\end{construction}

The next lemma shows some useful relations between the operators $T$ and $\widetilde{T}.$

\begin{lemma}\label{lem:capacity-expansion} 
If $T: M_{n_1} (\mathbb{C}) \to M_{n_2}(\mathbb{C})$ is an operator and $\widetilde{T}$ is the
operator obtained from Construction~\ref{con:squaring-operator}, the following hold:
\begin{enumerate}
\item $\capac(\widetilde{T}) = \capac(T)^{n_1}$.
\item $T$ is fractional-rank non-decreasing iff $\widetilde{T}$ is rank non-decreasing.
\end{enumerate}
\end{lemma}

\begin{proof}
\begin{align*}
\capac(\widetilde{T}) &= \inf \left\{ \Det \left(\widetilde{T}(X) \right) : X \succ 0, \Det(X) = 1 \right\} \\
&= \inf \left\{ \Det \left( \frac{1}{n_1} \sum_{i=1}^{n_2} T(X_{i,i}) \right)^{n_1} : X \succ 0, \Det(X) = 1 \right\} \\
&= \inf \left\{ \Det \left( \frac{n_2}{n_1} \cdot \frac{1}{n_2} \sum_{i=1}^{n_2} T(X_{i,i}) \right)^{n_1} : X_{i,i} \succ 0 \: \text{for all $i$}, \prod_{i=1}^{n_2} \Det(X_{i,i}) = 1 \right\} \\
&= \inf \left\{ \Det \left( \frac{n_2}{n_1} T(Y) \right)^{n_1} : Y \succ 0, \Det(Y) = 1 \right\} \\
&= \capac(T)^{n_1}
\end{align*}
The second equality follows from a generalization of Hadamard's inequality, determinant of a block diagonal psd matrix is smaller than the product of determinants of its diagonal blocks: 
$$
\Det(X) \le \prod_{i=1}^{n_2} \Det(X_{i,i})
$$
The third equality follows from replacing $\frac{1}{n_2} \sum_{i=1}^{n_2} X_{i,i}$ by $Y$ and using the log-concavity of determinant:
$$
\Det(Y) = \Det \left( \frac{1}{n_2} \sum_{i=1}^n X_{i,i} \right) \ge \prod_{i=1}^{n_2} \Det(X_{i,i})^{1/n_2} = 1
$$ 
Now let us move on to the second part of the lemma. Suppose $T$ is fractional-rank decreasing. So there exists $Y \succeq 0$ s.t. 
$$
\frac{\text{Rank}(T(Y))}{n_2} < \frac{\text{Rank}(Y)}{n_1}
$$
Let $X$ be the $n_2 \times n_2$ block diagonal matrix, all of whose diagonal blocks are $Y$. Then $\widetilde{T}(X)$ is a $n_1 \times n_1$ block diagonal matrix all of whose diagonal blocks are $\frac{n_2}{n_1} T(Y)$. Hence 
$$
\text{Rank}(\widetilde{T}(X)) = n_1 \Rk(T(Y)) < n_2 \Rk(Y) = \Rk(X)
$$
and thus $\widetilde{T}$ is rank decreasing. The other direction is more interesting. Suppose $\widetilde{T}$ is rank decreasing and it decreases the rank of $X$. Then we will prove that $T$ is fractional-rank decreasing and $\sum_{i=1}^{n_2} X_{i,i}$ is a witness i.e.
$$
\frac{\Rk \left(T \left(\sum_{i=1}^{n_2} X_{i,i}\right) \right)}{n_2} < \frac{\Rk \left( \sum_{i=1}^{n_2} X_{i,i}\right)}{n_1}
$$
Now 
\begin{align*}
n_1 \cdot \Rk \left(T \left(\sum_{i=1}^{n_2} X_{i,i}\right) \right) &= \Rk \left( \widetilde{T}(X) \right) \\
&< \Rk(X)
\end{align*}
So we would be done if we can prove that 
$$
\Rk(X) \le n_2 \cdot \Rk \left( \sum_{i=1}^{n_2} X_{i,i}\right)
$$
This follows from the following chain of inequalities:
$$
\Rk(X) \le \sum_{i=1}^{n_2} \Rk(X_{i,i}) \le n_2 \cdot \max_{\text{$i=1$ to $n_2$}} \Rk(X_{i,i}) \le n_2 \cdot \Rk \left(\sum_{i=1}^{n_2} X_{i,i} \right)
$$
The last inequality is true because for psd matrices $Y_1,Y_2$, $\Rk(Y_1 + Y_2) \ge \Rk(Y_1)$ and also the fact that $X_{i,i}$ are psd matrices since $X$ is psd. The second inequality is obvious. The first inequality follows from the following claim:
\begin{claim}
Let 
\[
X=
  \begin{bmatrix}
    X_{1,1} & X_{1,2}  \\
    X_{2,1} & X_{2,2} 
  \end{bmatrix}
\]
be a psd matrix s.t. $X_{1,1}$ and $X_{2,2}$ are square matrices of possibly different dimension. Then $\Rk(X) \le \Rk(X_{1,1}) + \Rk(X_{2,2})$. 
\end{claim}

\begin{proof}
We will prove that $\text{dim}(\text{ker}(X_{1,1})) + \text{dim}(\text{ker}(X_{2,2})) \le \text{dim}(\text{ker}(X))$ from which the claim easily follows. Let $v \in \text{ker}(X_{1,1})$. Then look at the vector 
\[v' =\begin{bmatrix}
         v \\
         0
        \end{bmatrix}
\]
Since $X_{1,1} v = 0$, we get that $v'^{\dagger} X v' = 0$. Since $X$ is psd, this implies that $v' \in \text{ker}(X)$. Similarly, if  $w \in \text{ker}(X_{2,2})$, the vector 
\[w' =\begin{bmatrix}
         0 \\
         w
        \end{bmatrix}
\]
lies in $\text{ker}(X)$. The proof is finished by observing that $v'$ and $w'$ lie in orthogonal subspaces. 
\end{proof}
\end{proof}

Now we will see that with Construction \ref{con:squaring-operator}, most of the results about square operators carry over to the rectangular case.

\begin{corollary}{\label{cap>0:expanding}} Let $T : M_{n_1}(\mathbb{C}) \rightarrow M_{n_2}(\mathbb{C})$ be a completely positive operator. Then $T$ is fractional-rank non-decreasing iff $\capac(T) > 0$. 
\end{corollary}

\begin{proof}
Immediate from Theorem \ref{thm:cap_RND} and Lemma \ref{lem:capacity-expansion}.
\end{proof}

For any completely positive operator, we have that $\capac(T) > 0$ iff $\capac(T^*) > 0$. This is proved below:

\begin{proposition}\label{prop:cap-dual} Let $T : M_{n_1}(\mathbb{C}) \rightarrow M_{n_2}(\mathbb{C})$ be a completely positive operator. Then $\capac(T)^{1/n_2} = \frac{n_2}{n_1} \capac(T^*)^{1/n_1}$.
\end{proposition}

\begin{proof}
$$
\capac(T)^{1/n_2} = \text{inf}_{X \succ 0} \left\{ \frac{\Det\left( \frac{n_2}{n_1} T(X) \right)^{1/n_2}}{\Det(X)^{1/n_1}} \right\}
$$
and 
$$
\frac{1}{\capac(T^*)^{1/n_1}} = \text{sup}_{Y \succ 0} \left\{ \frac{\Det(Y)^{1/n_2}}{\Det\left( \frac{n_1}{n_2} T^*(Y) \right)^{1/n_1}} \right\}
$$
Thus
\begin{align*}
\frac{\capac(T)^{1/n_2}}{\capac(T^*)^{1/n_1}} &= \text{inf}_{X \succ 0} \: \text{sup}_{Y \succ 0} \left\{ \frac{\Det \left( \frac{n_2}{n_1} T(X)\cdot Y \right)^{1/n_2}}{\Det \left( \frac{n_1}{n_2}X \cdot T^*(Y) \right)^{1/n_1}} \right\} \\
&\ge \text{inf}_{X \succ 0} \left\{ \frac{\Det \left( \frac{n_2}{n_1} I_{n_2} \right)^{1/n_2}}{\Det \left( \frac{n_1}{n_2} X \cdot T^* \left( T(X)^{-1}\right) \right)^{1/n_1}} \right\} \\
&\ge \frac{n_2}{n_1} \cdot \text{inf}_{X \succ 0} \left\{  \frac{n_1}{\tr \left[  \frac{n_1}{n_2} X \cdot T^* \left( T(X)^{-1}\right)\right]}  \right\} \\
&= \frac{n_2}{n_1} \cdot \text{inf}_{X \succ 0} \left\{  \frac{n_1}{\tr \left[  \frac{n_1}{n_2} T(X) \cdot  T(X)^{-1} \right]}  \right\} \\
&= \frac{n_2}{n_1}
\end{align*}
The first inequality follows by taking $Y = T(X)^{-1}$. Second inequality follows from AM-GM. Second equality follows from the fact that $\tr[A \cdot T^*(B)] = \tr[T(A) \cdot B]$ for all matrices $A,B$ of appropriate sizes. We have proved that 
$$
\frac{\capac(T)^{1/n_2}}{\capac(T^*)^{1/n_1}} \ge \frac{n_2}{n_1}
$$
By applying the same argument to $T^*$ instead of $T$, we get that
$$
\frac{\capac(T^*)^{1/n_1}}{\capac(T)^{1/n_2}} \ge \frac{n_1}{n_2}
$$
Combining these completes the proof. 
\end{proof}

Thus, we can deduce the following corollary:

\begin{corollary}\label{cor:dual-expanding}
	For any completely positive operator $T : M_{n_1}(\mathbb{C}) \rightarrow M_{n_2}(\mathbb{C})$,
	$\capac(T) > 0$ iff $T^*$ is fractional-rank non-decreasing.
\end{corollary}

Similar to the square case, whenever a rectangular operator $T$ is fractional-rank non-decreasing, there exists an operator scaling which is almost doubly stochastic. 

\begin{corollary} Let $T : M_{n_1}(\mathbb{C}) \rightarrow M_{n_2}(\mathbb{C})$ be a completely positive operator s.t. $\capac(T) > 0$. Then for every $\eps > 0$, there exists an operator scaling $T'$ of $T$ s.t. $\ds(T') \le \eps$. Furthermore this $T'$ can be found by running Algorithm \ref{Gurvits_alg_rect} for an appropriate number of steps $t(T, \eps)$. 
\end{corollary}

\begin{proof}
Given a $T : M_{n_1}(\mathbb{C}) \rightarrow M_{n_2}(\mathbb{C})$, let us construct a $\widetilde{T} : M_{n_1 n_2}(\mathbb{C}) \rightarrow M_{n_1 n_2}(\mathbb{C})$ via Construction \ref{con:squaring-operator}. Since $\capac(T) > 0$, by Lemma \ref{lem:capacity-expansion}, $\capac(\widetilde{T}) > 0$. Then by Theorem \ref{thm:cap_scaling}, if we run Algorithm $G$ on $\widetilde{T}$ for an appropriate number of steps $t(\widetilde{T}, \eps')$, the output $\widetilde{Z} = \widetilde{T}_{t(\widetilde{T}, \eps')}$ would satisfy $\ds(\widetilde{Z}) \le \eps'$. It is not hard to see that $\widetilde{Z}$ is of the form 
$$
\widetilde{Z}(X) = I_{n_1} \otimes \left( \frac{1}{n_1} \sum_{i=1}^{n_2} Z(X_{i,i})\right)
$$
where $Z$ is the result of running Algorithm \ref{Gurvits_alg_rect} on $T$ for $t(\widetilde{T}, \eps')$ steps. 
Also $\ds(\widetilde{Z}) = n_1 \ds(Z)$, so choosing $\eps' = \eps/n_1$ proves that running Algorithm \ref{Gurvits_alg_rect} on $T$ for $t(\widetilde{T}, \eps'/n_1)$ steps outputs $T'$ s.t. $\ds(T') \le \eps$. 
\end{proof}

Algorithms for computing capacity of square operators can be used in a black box way to compute the capacity of rectangular operators because of Construction \ref{con:squaring-operator} and Lemma \ref{lem:capacity-expansion}.

\begin{corollary}\label{cor:GGOW} Let $T : M_{n_1}(\mathbb{C}) \rightarrow M_{n_2}(\mathbb{C})$ be a completely positive operator s.t. the bit-complexity of the description of $T$ $($in terms of its Kraus operators$)$ is $b$. Then there is a $\poly(n_1,n_2,b)$ time $($deterministic$)$ algorithm to test whether $T$ is fractional-rank non-decreasing. There is also a $\poly(n_1,n_2,,b,1/\eps)$ time algorithm to compute a $(1+\eps)$-multiplicative approximation to $\capac(T)$. Furthermore, the algorithm outputs a scaling $T'$ of $T$ which is almost doubly stochastic i.e. $\capac(T') \ge 1-\eps$. 
\end{corollary}

Similarly, the algorithm in Theorem \ref{alg:IQS} \cite{IQS15b} for computing witnesses to rank decreasing property of square operators can be used in a black box way to compute witnesses to fractional-rank decreasing property of rectangular operators.

\begin{corollary}\label{cor:IQS_rect} Let $T : M_{n_1}(\mathbb{C}) \rightarrow M_{n_2}(\mathbb{C})$ be a completely positive operator s.t. the bit-complexity of the description of $T$ $($in terms of its Kraus operators$)$ is $b$. Then there is a $\poly(n_1, n_2, b)$ time $($deterministic$)$ algorithm to test whether $T$ is fractional-rank non-decreasing. Furthermore, if $T$ is fractional-rank decreasing, then the algorithm also outputs an $X \succeq 0$ s.t. $\frac{\text{Rank}(T(X))}{n_2} <  \frac{\text{Rank}(X)}{n_1}$.
\end{corollary}

\section{Reduction from Brascamp-Lieb to completely positive operators}\label{sec:bl-to-cap}

In this section we will show how each Brascamp-Lieb datum $(\mathbf{B, p})$ corresponds 
to an operator $T_{(\mathbf{B,p})}$ for which 
$\capac(T_{(\mathbf{B,p})}) = 1/\BL(\mathbf{B, p})^2$. From this correspondence, and from 
the connection between capacity and rank expansion of an operator, we will show how to
derive the~\cite{BCCT} conditions for finiteness of $\BL(\mathbf{B, p})$. Recall that in the
Brascamp-Lieb setting, we are given a datum $(\mathbf{B, p})$ consisting of $m$ matrices 
$B_1,\ldots,B_m$, where each $B_i$ is a $n_i \times n$ matrix. We also have $m$ 
non-negative rational numbers $p_1, \ldots, p_m$ s.t. $p_i = c_i/d$ for some
$c_i, d \in \N$. Since one necessary condition for finiteness of the Brascamp-Lieb constant
is that $n = \sum_{i=1}^m p_i n_i$, we will assume this equality throughout the 
section.

We begin by recalling the characterization of the Brascamp-Lieb constant due to Lieb \cite{Lieb90}:

\begin{theorem}[\textbf{Brascamp-Lieb constant}] Given a Brascamp-Lieb datum 
$(\mathbf{B, p})$, the Brascamp-Lieb constant $\BL(\mathbf{B, p})$ equals:
$$
\BL(\mathbf{B, p}) = \sup \left\{ \left(
\frac{\prod_{i=1}^m \Det(X_i)^{p_i}}{\Det \left( \sum_{i=1}^m p_i B_i^{\dagger} X_i B_i \right)} 
\right)^{1/2} : X_i \succ 0 \: \text{is a $n_i \times n_i$ matrix} \right\}
$$
\end{theorem}



We now describe our main reduction from the BL datum to (rectangular) completely positive operators. 

\begin{construction}\label{cons:main_reduction}
Let $m' = \sum_{i=1}^m c_i$. We will construct $m'$ $nd \times n$ matrices 
$A_1,\ldots,A_{m'}$ from $B_1,\ldots,B_m$. Essentially $A_1,\ldots,A_{m'}$ will consist of 
$c_i$ copies of $B_i$ acting on different parts of the space. Consider the following matrix $A$:
\[ 
A = 
\begin{pmatrix}
	\ B_1 \ \\
	\ \vdots \ \\
	\ B_1 \ \\
	\ B_2 \ \\
	\ \vdots \ \\
	\ B_2 \ \\
	\\
	\ \vdots \ \\
	\\
	\ B_m \ \\
	\ \vdots \ \\
	\ B_m \ \\
\end{pmatrix}
\begin{matrix}
	\begin{rcases*}
		\\
		\\
		\\
	\end{rcases*} c_1 \text{ times }  \\
	\begin{rcases*}
		\\
		\\
		\\
	\end{rcases*} c_2 \text{ times } \\
	\\
	\vdots \\
	\\
	\begin{rcases*}
		\\
		\\
		\\
	\end{rcases*} c_m \text{ times }	
\end{matrix}
\]


$A$ has a natural block structure consisting of a total of $m'$ blocks. The matrices $A_j$ are obtained 
by keeping one block intact and zeroing out all other blocks. More formally, suppose $j$ is such 
that $\sum_{i=1}^k c_i < j \le \sum_{i=1}^{k+1} c_i$. Then all entries of $A_j$ will be 
zero except rows 
$$
\sum_{i=1}^k c_i n_i + \left(j - \sum_{i=1}^k c_i - 1\right) \cdot n_{k+1} + 1
\text{ to }
\sum_{i=1}^k c_i n_i + \left(j - \sum_{i=1}^k c_i - 1\right) \cdot n_{k+1} + n_{k+1}
$$ 
which are occupied by $B_{k+1}$.
Define the operator $T_{(\mathbf{B, p})} : M_{nd}(\mathbb{C}) \rightarrow M_{n}(\mathbb{C})$ as 
$$
T_{(\mathbf{B, p})}(X) = \sum_{j=1}^{m'} A_j^{\dagger} X A_j
$$
\end{construction}

\begin{definition}[\textbf{Brascamp-Lieb operator}]
	The completely positive operator $T_{(\mathbf{B, p})}$ constructed in Construction~\ref{cons:main_reduction}
	is called the {\em Brascamp-Lieb operator} for datum $(\mathbf{B, p}).$
\end{definition}

The next lemma relates important properties of a BL datum and its Brascamp-Lieb operator. 

\begin{lemma}\label{lem:main_reduction}
Let $(\mathbf{B, p})$ be a Brascamp-Lieb datum, where 
	$p_i = c_i/d$ for $c_i, d \in \N$ and $B_i : \R^n \to \R^{n_i}$. Moreover, assume that
	$n = \sum_{i=1}^m c_i n_i/d.$ Let $T_{(\mathbf{B, p})}$ be the Brascamp-Lieb operator corresponding to 
	$(\mathbf{B, p})$. Then
	\begin{enumerate}
	\item If $\BL(\mathbf{B, p}) < \infty$, then $$\capac(T_{(\mathbf{B, p})}) = 1/\BL(\mathbf{B, p})^2$$
	\item  The inequality 
	$$
	\dim (V) \leq \sum_j p_j  \dim(B_jV)
	$$
	holds for all subspaces $V \subseteq \R^n$ iff $T_{(\mathbf{B, p})}^*$ is fractional-rank non-decreasing. 
	\end{enumerate}
	Moreover, if $B_i$'s are integer matrices such that $\| B_i \|_\infty \le M$, 
	then the same holds for the matrices $A_1, \ldots, A_{m'}$. 
\end{lemma}

\begin{proof}
Let us start by proving the first part of the lemma. 
For $j$ s.t. $\sum_{i=1}^k c_i < j \le \sum_{i=1}^{k+1} c_i$ (let us denote this by $i_j = k+1$), 
lets denote by $X_j$ the $n_{k+1} \times n_{k+1}$ submatrix of $X$ indexed by rows 
and columns from 
$$
\sum_{i=1}^k c_i n_i + \left(j - \sum_{i=1}^k c_i - 1\right) \cdot n_{k+1} + 1
\text{ to }
\sum_{i=1}^k c_i n_i + \left(j - \sum_{i=1}^k c_i - 1\right) \cdot n_{k+1} + n_{k+1}
$$ 
Then it can be seen that 
$$
T_{(\mathbf{B, p})}(X) = \sum_{j=1}^{m'} B_{i_j}^{\dagger} X_j B_{i_j}
$$
Now let us calculate the capacity of $T_{(\mathbf{B, p})}$:
\begin{align*}
\capac(T_{(\mathbf{B, p})}) &= \inf \left\{ \Det\left( \frac{T_{(\mathbf{B, p})}(X)}{d} \right) : X \succ 0, \Det(X) = 1\right\} \\
&= \inf \left\{ \Det\left( \frac{\sum_{j=1}^{m'} B_{i_j}^{\dagger} X_j B_{i_j}}{d} \right) : X \succ 0, \Det(X) = 1\right\} \\
&= \inf \left\{ \Det\left( \frac{\sum_{j=1}^{m'} B_{i_j}^{\dagger} X_j B_{i_j}}{d} \right) : X_{j} \succ 0, \prod_{j=1}^{m'} \Det(X_j) = 1, X_j \: \text{is }n_{i_j} \times n_{i_j} \right\} \\
&= \inf \left\{ \Det\left( \frac{\sum_{i=1}^{m} c_i B_{i}^{\dagger} Y_i B_{i}}{d} \right) : Y_{i} \succ 0, \prod_{i=1}^{m} \Det(Y_i)^{c_i} = 1, Y_i \: \text{is }n_{i} \times n_{i} \right\} \\
&= \inf \left\{ \frac{\Det\left( \frac{\sum_{i=1}^{m} c_i B_{i}^{\dagger} Y_i B_{i}}{d} \right)}{\prod_{i=1}^m \Det(Y_i)^{c_i/d}} : Y_{i} \succ 0, Y_i \: \text{is }n_{i} \times n_{i} \right\} \\
&= \inf \left\{ \frac{\Det\left(\sum_{i=1}^{m} p_i B_{i}^{\dagger} Y_i B_{i} \right)}{\prod_{i=1}^m \Det(Y_i)^{p_i}} : Y_{i} \succ 0, Y_i \: \text{is }n_{i} \times n_{i} \right\} \\
&= 1/\BL(\mathbf{B, p})^2
\end{align*}
For third equality, we use the generalization of Hadamard's inequality. For fourth equality, we 
use log-concavity of determinant after replacing $\sum_{j : i_j = k} X_{j}$ by $c_k Y_k$. 

Now we move on to prove the second part of the lemma. 
The condition that the dual operator $T_{(\mathbf{B, p})}^* : \cM_n(\C) \to \cM_{nd}(\C)$ is fractional-rank 
	non-decreasing is the following:
	$$ \Rk(T_{(\mathbf{B, p})}^*(X))/d \ge \Rk(X), \ \text{ for all } X \succeq 0$$
	
	As the matrices $A_j$ defined in Construction~\ref{cons:main_reduction} correspond to distinct 
	copies of $B_j$ acting on $X$, the dual operator corresponds to the following 
	transformation:
	\begin{equation}\label{eq:dir-sum}
	T_{(\mathbf{B, p})}^*(X) = \bigoplus_{j=1}^m 
	\left( \bigoplus_{i=1}^{c_j} B_i X B_i^\dagger \right)
	\end{equation}
	Hence, the condition $\Rk(T^*(X))/d \ge \Rk(X)$ becomes:
	$$   \Rk(X) \leq \Rk(T^*(X))/d = \frac{1}{d} \cdot \sum_{j=1}^m c_j \Rk(B_j X B_j^\dagger)
	= \sum_{j=1}^m p_j \Rk(B_j X B_j^\dagger). $$
	Let $V = \langle \bv_1, \ldots, \bv_k \rangle$ be an orthonormal eigenvector basis for $X$, 
	assuming $\Rk(X) = k$. Hence, we have $\lambda_i > 0$ such that
	$X = \sum_{i=1}^k \lambda_i \bv_i \bv_i^\dagger$ and the inequality above becomes:
	\begin{align*}
		\dim(V) = \Rk(X) &\leq \sum_{j=1}^m p_j \Rk(B_j X B_j^\dagger) \\
		&= \sum_{j=1}^m p_j \Rk\left(B_j 
		\left( \sum_{i=1}^k  \lambda_i \bv_i \bv_i^\dagger \right)  B_j^\dagger \right) \\ 
		&= \sum_{j=1}^m p_j \dim(B_j V).
	\end{align*}	   
	As this inequality must hold for any $X \succeq 0$, it must be the case that it holds
	for all spaces $V \subseteq \mathbb{R}^n$. This is precisely the condition of~\cite{BCCT}. 

Notice that by our construction of the matrices $A_j$, if $B_i$'s are integer matrices with 
$\| B_i \|_\infty \le M$, then matrices $A_j$ will also be integer matrices whose entries 
are bounded by $M$.
\end{proof}

Now we will see that with Construction \ref{cons:main_reduction} and Lemma \ref{lem:main_reduction}, we can deduce the following corollaries for the BL datum. The next corollary is Theorem \ref{feasibility} from the introduction proved in \cite{BCCT}. 

\begin{corollary}\label{cor:BCCT} Let $(\mathbf{B, p})$ be a Brascamp-Lieb datum, where 
	$p_i = c_i/d$ for $c_i, d \in \N$ and $B_i : \R^n \to \R^{n_i}$. Moreover, assume that
	$n = \sum_{i=1}^m c_i n_i/d.$ Then $\BL(\mathbf{B, p})$ is finite iff for every subspace 
$V \subseteq \mathbb{R}^n$, it holds that
$$
\dim(V) \le \sum_{i=1}^m p_i \dim(B_i(V))
$$
\end{corollary}

\begin{proof}
Follows immediately from Corollary \ref{cor:dual-expanding} and Lemma \ref{lem:main_reduction}. 
\end{proof}

Next we will derive an algorithm for checking finiteness of the Brascamp-Lieb constant from Construction \ref{cons:main_reduction} and Corollary~\ref{cor:IQS_rect}. 

\begin{corollary}\label{cor:bl_feasibility} Let $(\mathbf{B, p})$ be a Brascamp-Lieb datum, where 
	$p_i = c_i/d$ for $c_i, d \in \N$ and $B_i : \R^n \to \R^{n_i}$. Moreover, assume that
	$n = \sum_{i=1}^m c_i n_i/d.$ Let $b$ be the total bit size of $\mathbf{B}$. Then there is a $\poly(b,d)$ time 
	$($deterministic$)$ algorithm to check if $\BL(\mathbf{B, p})$ is finite. Furthermore, if $\BL(\mathbf{B, p})$ is not finite, then
	the algorithm outputs a subspace $V \subseteq \mathbb{R}^n$ s.t.
	$$
\dim(V) > \sum_{i=1}^m p_i \dim(B_i(V))
$$
\end{corollary}

\begin{proof}
	Let $T_{(\mathbf{B, p})} : \cM_{nd}(\C) \to \cM_n(\C)$ be the operator defined 
	in Construction~\ref{cons:main_reduction} and let
	$T_{(\mathbf{B, p})}^* : \cM_n(\C) \to \cM_{nd}(\C)$ be its dual. By 
	Lemma~\ref{lem:main_reduction} and Proposition~\ref{prop:cap-dual}, we have that 
	$\BL(\mathbf{B, p})$ is finite iff $T_{(\mathbf{B,p})}^*$ is fractional-rank non-decreasing. 
	Corollary~\ref{cor:IQS_rect} gives us a deterministic $\poly(b,d)$-time algorithm to test
	whether the operator $T_{(\mathbf{B,p})}^*$ is fractional-rank non-decreasing, and if it is
	not, the algorithm outputs $X \succeq 0$ such that 
	\begin{equation}\label{eq:frac-rk-ineq}
		\dfrac{\rk(T(X))}{nd} < \dfrac{\rk(X)}{n}
	\end{equation} 
	In the latter case, to find the subspace $V$
	which violates the dimension inequality, let $V = \langle \bv_1, \ldots, \bv_k \rangle$ be
	an orthonormal eigenvector basis for $X$. Hence, we have 
	$X = \dst\sum_{i=1}^k \lambda_i \bv_i \bv_i^\dagger$. According to 
	equation~\eqref{eq:dir-sum},
	$$ T_{(\mathbf{B, p})}^*(X) = \bigoplus_{j=1}^m 
	\left( \bigoplus_{i=1}^{c_j} B_i X B_i^\dagger \right) $$
	Therefore, from~\eqref{eq:frac-rk-ineq} we obtain:
	\begin{align*}
		\dfrac{\rk(X)}{n} > \dfrac{\rk(T(X))}{nd} \then
		\dim(V) = \rk(X) &> \dfrac{\rk(T(X))}{d} \\
		&= \sum_{j=1}^m p_j \Rk(B_j X B_j^\dagger) \\
		&= \sum_{j=1}^m p_j \Rk\left(B_j 
		\left( \sum_{i=1}^k  \lambda_i \bv_i \bv_i^\dagger \right)  B_j^\dagger \right) \\ 
		&= \sum_{j=1}^m p_j \dim(B_j V).
	\end{align*}
	The above proves that $Im(X)$ is a vector space which violates the dimension inequalities. 
\end{proof}

\subsection{Connection to quiver semi-invariants}\label{sec:quiver}
In this short subsection, we provide an invariant theoretic view of our reduction\footnote{We remark that other connections between BL inequalities and quiver representations are mentioned in \cite{BCCT}.}. It is far from being self-contained and we refer the reader to the survey \cite{DW2000} about semi-invariants of quiver representations. In this theory, there is a general way of generating semi-invariants for general quivers by a reduction to the Kronecker quiver (left-right action) (\cite{DW2000, DZ2001, SVdB2001}, also see Section $5$ in \cite{DM15}). 
As it happens, the data to both BL inequalities and to completely positive operators fits the descriptions of the ``star quiver''and the ``Kronecker quiver'' respectively, and with hindsight,  the reduction described in Construction~\ref{cons:main_reduction} (when composed with the reduction from rectangular to square operators) can be viewed as a special case of the reduction above. More specifically, semi-invariants for the star quiver with dimensions $(n, n_1, n_2, \ldots, n_m)$ and weight vector $(d, -c_1,-c_2,\ldots,-c_m)$ can be generated by reducing to a Kronecker quiver of appropriate dimensions (see Figure \ref{fig:star_to_Kronecker}). It is worth emphasizing here the intriguing connection (that this paper exploits) between analysis (BL inequalities) and algebra (semi-invariants of the star quiver): a BL datum $(B_1,\ldots,B_m, c_1/d,\ldots,c_m/d)$ is not feasible iff $(B_1,\ldots,B_m)$ is the common zero of the semi-invariants of the star quiver with dimensions $(n, n_1, n_2, \ldots, n_m)$ and weight vectors \linebreak $\{ (l\cdot d, -l \cdot c_1,-l \cdot c_2,\ldots,-l \cdot c_m):  l \in \mathbb{N}\}$\footnote{One can even define such a connection for the finiteness of a specific BL constant: $\BL(\mathbf{B,p}) = \infty$ iff $\mathbf{B}$ is in the null cone of the subgroup of $\text{GL}_n(\C) \times \text{GL}_{n_1}(\C) \times \cdots \text{GL}_{n_m}(\C)$ satisfying $\Det(C) = \prod_{i=1}^m \Det(C_i)^{p_i}$.}.

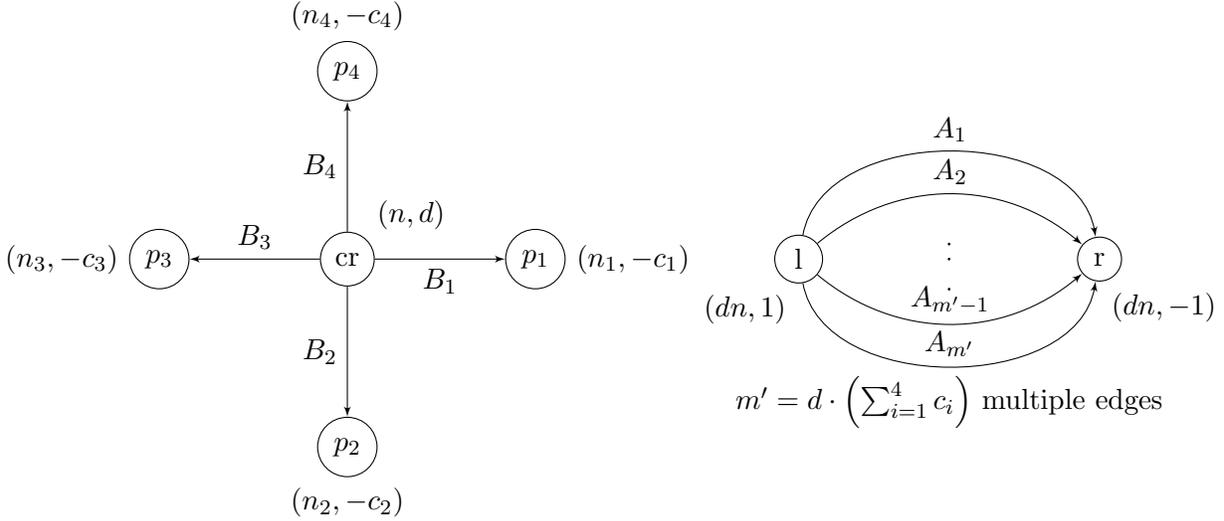
\begin{figure}[H]
\centering
\begin{tikzpicture}

\tikzset{vertex/.style = {shape=circle,draw,minimum size=1.5em}}
\tikzset{edge/.style = {->,> = latex'}}
\node[vertex] (cr) at  (0,0) [label={[label distance=0.01cm]45:$(n,d)$}] {cr};
\node[vertex] (p_1) at  (2.5,0) [label={[label distance=0.01cm]0:$(n_1,-c_1)$}] {$p_1$};
\node[vertex] (p_2) at  (0,-2.5) [label={[label distance=0.01cm]-90:$(n_2,-c_2)$}] {$p_2$};
\node[vertex] (p_3) at  (-2.5,0) [label={[label distance=0.01cm]180:$(n_3,-c_3)$}] {$p_3$};
\node[vertex] (p_4) at  (0,2.5)  [label={[label distance=0.01cm]90:$(n_4,-c_4)$}]  {$p_4$};

\draw[edge] (cr) -- node[below] {$B_1$} ++ (p_1);
\draw[edge] (cr) -- node[left] {$B_2$} ++ (p_2);
\draw[edge] (cr) -- node[above] {$B_3$} ++ (p_3);
\draw[edge] (cr) -- node[left] {$B_4$} ++ (p_4);

\node[vertex] (l) at  (6,0) [label={[label distance=0.01cm]-95:$(dn,1)$}] {l};
\node[vertex] (r) at  (10,0) [label={[label distance=0.01cm]-85:$(dn,-1)$}] {r};

\draw[edge][yshift=0.5cm]  (l)  to [bend right=40] node[sloped,midway,above] {$A_{m'-1}$} (r);
\draw[edge][yshift=0.5cm]  (l)  to [bend left=40] node[sloped,midway,above] {$A_2$} (r);

\draw[edge][yshift=0.5cm]  (l)  to [bend right=80] node[sloped,midway,above] {$A_{m'}$} (r);
\draw[edge][yshift=0.5cm]  (l)  to [bend left=80] node[sloped,midway,above] {$A_1$} (r);

\node[fill,circle,scale=0.1] at (8,0) {};
\node[fill,circle,scale=0.1] at (8,0.2) {};
\node[fill,circle,scale=0.1] at (8,-0.4) [label={[label distance=1cm]-90:$m'=d \cdot \left( \sum_{i=1}^4 c_i\right)$ multiple edges}] {};

\end{tikzpicture}
\caption{Reduction from star quiver to the Kronecker quiver. Nodes are labelled with $(d,w)$: $d$ is the dimension of the vector space at this node and $w$ is the weight corresponding to this node. Edges are labelled with maps from the vector space at the head of the edge to the vector space at the tail of the edge.} \label{fig:star_to_Kronecker}
\end{figure}

\section{Upper bound on the Brascamp-Lieb constant}\label{sec:ub_bl}

In this section, we will prove an exponential upper bound on the BL constant $\BL(\mathbf{B}, \mathbf{p})$ (Theorem \ref{thmintro:BL_ub}). The main ingredient is a  lower bound on the capacity of square completely positive operators with integer entries from \cite{GGOW}.

\begin{theorem}[\textbf{Capacity of Square Operators}, Theorem 2.18 in \cite{GGOW}]\label{lem:cap-square}
	Suppose $T_A$ is a completely positive operator which has positive capacity and has Kraus operators 
	$A_1, \ldots, A_m \in \cM_n(\Z)$. In this case:
	$$ \capac(T_A) \ge  \expon(-2 n \log(n)) $$
\end{theorem}

We now derive the upper bound on the BL-constant from the lower bound on capacity above via our main reduction. We do so first for BL-data with integer entries (which is independent of their magnitude!), and then for general ones with rational entries.

\begin{theorem}[\textbf{Lower Bound on Brascamp-Lieb Capacity}]\label{thm:lb-bl-capac}
	Let $(\mathbf{B, p})$ be a Brascamp-Lieb datum, with $B_i \in \cM_{n_i \times n}(\Z)$
	and $p_i = c_i / d$, where $c_i, d \in \N$. Let $T_{(\mathbf{B, p})}$ be the Brascamp-Lieb operator, as defined 
	in Construction~\ref{cons:main_reduction}. If $T_{(\mathbf{B, p})}$ has positive capacity, 
	then the following bound holds:
	$$ \capac(T_{(\mathbf{B, p})})  \geq  \exp(-2n\log(n^2 d)).$$
\end{theorem}

\begin{proof}
	According to Construction~\ref{cons:main_reduction}, we construct an operator 
	$T_{(\mathbf{B, p})} : \cM_{nd}(\C) \to \cM_n(\C)$ such that
	$\capac(T_{(\mathbf{B, p})}) = 1/\BL(\mathbf{B, p})^2$ and
	$T_{(\mathbf{B, p})}(X) = \dst\sum_{j=1}^{m'} A_j^{\dagger} X A_j$, where $m' \le nd$, 
	each $A_j$ has integer coordinates. From operator $T_{(\mathbf{B, p})}$, construct operator 
	$\widetilde{T}_{(\mathbf{B, p})} : \cM_{n^2d} \to \cM_{n^2d}$ according
	to Construction~\ref{con:squaring-operator}. Item 1 of 
	Lemma~\ref{lem:capacity-expansion} implies that 
	$$ \capac(T_{(\mathbf{B, p})}) = \capac(\widetilde{T}_{(\mathbf{B, p})})^{1/nd}. $$

	Moreover, item 1 of Lemma~\ref{lem:capacity-expansion}, together with the assumption
	that $T_{(\mathbf{B, p})}$ has positive capacity, implies that $\widetilde{T}_{(\mathbf{B, p})}$
	has positive capacity. Hence, Theorem~\ref{lem:cap-square} implies that 
	$$ \capac(\widetilde{T}_{(\mathbf{B, p})}) \ge \exp(-2 n^2 d \log(n^2 d)).$$
	Thus, we have:
	$$ \capac(T_{(\mathbf{B, p})}) = 
	\capac(\widetilde{T}_{(\mathbf{B, p})})^{1/nd} \ge \exp(-2n\log(n^2 d)). $$	
\end{proof}

From the theorem above we obtain the following upper bound on the Brascamp-Lieb 
constant. Note that we get a bound which is independent of the magnitude of the entries of $A_i$'s. This is possible because we restrict $A_i$'s to be integer matrices.

\begin{corollary}[\textbf{Upper Bound on Brascamp-Lieb Constant}]\label{cor:bl_upperbound}
	Let $(\mathbf{B, p})$ be a Brascamp-Lieb datum, with $B_i \in \cM_{n_i \times n}(\Z)$
	and $p_i = c_i / d$, where $c_i, d \in \N$. If $\BL(\mathbf{B, p}) < \infty$, then the following bound holds:
	$$ \BL(\mathbf{B, p}) \le \exp(4n\log(n^2 d)).$$
\end{corollary}

As a corollary, we can also get an exponential bound that is independent of $d$. 

\begin{corollary}\label{cor:bl_upperbound_ind}
	Let $(\mathbf{B, p})$ be a Brascamp-Lieb datum, with $B_i \in \cM_{n_i \times n}(\Z)$
	and $p_1,\ldots,p_m$ arbitrary positive reals. If $\BL(\mathbf{B, p}) < \infty$, then the following bound holds:
	$$ \BL(\mathbf{B, p}) \le \exp(12nm\log(mn)).$$
\end{corollary}

\begin{proof}
It is proven in \textnormal{\cite{CDKSY15}} that $\log(\BL(\mathbf{B, p}))$ is convex in $\mathbf{p}$. This can also be seen easily from the entropy formulation of Brascamp-Lieb inequalities \cite{Carlen09, Liu16}. Thus to upper bound, $\BL(\mathbf{B, p})$, it is enough to look at the vertices of the BL polytope $P_{\mathbf{B}}$ (from the log-convexity, it follows that if $\BL(\mathbf{B, p}) < \infty$ for some $\mathbf{p}$, then the maximum finite value will be achieved at the vertices of $P_{\mathbf{B}}$).  It is not hard to see that the vertices of the BL polytope $P_{\mathbf{B}}$ have common denominator $d \le m! \cdot n^m$. Plugging this into Corollary \ref{cor:bl_upperbound} gives us the desired bound. 
\end{proof}

From the bound on the BL constant for integer datum, we can also get a bound for arbitrary rational datum in terms of the bit sizes (which is a restatement of Theorem \ref{thmintro:BL_ub} in the introduction.).

\begin{corollary} \label{cor:bl_upperbound_rational}
	Let $(\mathbf{B, p})$ be a Brascamp-Lieb datum, with $B_i \in \cM_{n_i \times n}(\Q)$
	and $p_i = c_i / d$, where each entry of $B_i$ has bit-size at most $\widetilde{b}$ and $c_i, d \in \N$. If $\BL(\mathbf{B, p}) < \infty$, then the following bound holds:
	$$ \BL(\mathbf{B, p}) \le \exp(4n\log(n^2 d) + n \widetilde{b}).$$
\end{corollary}

\begin{proof}
Note that $\widetilde{B}_i = 2^{\widetilde{b}} B_i$ are integer matrices. Also it is easy to see that $\BL(\mathbf{B,p}) = \BL(\mathbf{\widetilde{B},p}) \cdot 2^{n \widetilde{b}}$. Then the desired bound follows from Corollary \ref{cor:bl_upperbound}.
\end{proof}

\begin{remark}
As before, $d$ can be eliminated from the above corollary if needed but the bound in the exponential would become quadratic.
\end{remark}

\section{Computing the Brascamp-Lieb constant}\label{sec:compute_bl}

In this section we show how Algorithm G with truncation, described in~\cite{gurvits2004, GGOW}, 
computes the Brascamp-Lieb constant up to a multiplicative factor of $(1 \pm \epsilon)$.
We begin by restating Theorem 3.5 from~\cite{GGOW}, which states that we can approximate
the capacity of a completely positive operator.

\begin{theorem}[Theorem 3.5 from~\cite{GGOW}]\label{thm:computing-capacity-truncated} 
Let $T$ be a completely positive operator, whose Kraus 
operators are given by $n \times n$ matrices $A_1, \ldots, A_m \in \cM_{n}(\Z)$, 
such that each entry of $A_i$ has absolute value at most $M$.  
Algorithm G, with truncation parameter 
$$
P(n, 1/\epsilon, \log(M)) = \frac{1}{\epsilon} \cdot (n^{12}\log^4(Mn)) \cdot \log(n^4/\epsilon^2)
$$ 
when applied for $t = \dfrac{4n^3}{\eps^2} \cdot \left(1 + 10n^2\log(Mn) \right)$ 
steps approximates $\capac(T)$ within a multiplicative factor of $1 \pm \epsilon$. 
\end{theorem}

Theorem~\ref{thm:computing-capacity-truncated} tells us how to approximate the capacity of an operator.
Construction~\ref{cons:main_reduction} gives us an operator whose capacity is the inverse square of 
the Brascamp-Lieb constant. Hence, to approximate the Brascamp-Lieb constant corresponding
to datum $(\mathbf{B, p})$, we need only obtain operator $T_{(\mathbf{B, p})}$ and use algorithm
$G$ to compute an approximation to $\capac(T_{(\mathbf{B, p})})$. Then, 
$\BL(\mathbf{B, p}) = \left( \dfrac{1}{\capac(T_{(\mathbf{B, p})})} \right)^{1/2}$.
This yields the following algorithm, and theorem (from which Theorem \ref{constant-alg} in the introduction follows
by a simple scaling argument):

\begin{Algorithm}
\textbf{Input}: Brascamp-Lieb datum $(\mathbf{B, p})$, where each $p_i = c_i/d$, with
$c_i, d \in \mathbb{N}$, $d \neq 0$, $B_i \in \cM_{n_i \times n}(\Z)$ and approximation parameter $\epsilon > 0$. 
Assume that each entry of $B_i$ has absolute value at most $M$. \\
\textbf{Output}: $\BL(\mathbf{B, p})$ with multiplicative error of $(1\pm\epsilon)$. 

\begin{enumerate}
	\item Check if $n = \dst\sum_{i=1}^m p_i n_i$. If yes, go to step 2. If not, return $\infty$.
	\item Let $T_{(\mathbf{B, p})}$ be the operator constructed in Construction~\ref{cons:main_reduction}.
		Let $\widetilde{T}_{(\mathbf{B, p})}$ be the operator obtained from $T_{(\mathbf{B, p})}$
		via Construction~\ref{con:squaring-operator}. Go to step 3.
	\item Apply Algorithm $G$ on $\widetilde{T}_{(\mathbf{B, p})}$, with approximation
	parameter $\epsilon$. Let $\alpha$ be the output of Algorithm $G$.
	\item Return $(1/\alpha)^{1/2nd}$.
\end{enumerate}
\caption{Algorithm $\widetilde{G}$ computing the BL constant}
\label{bl_alg}
\end{Algorithm}

\begin{theorem}[\textbf{Approximating the Brascamp-Lieb Constant}]\label{thm:computing_bl}
		Given a Brascamp-Lieb datum $(\mathbf{B, p})$, where $p_i = c_i/d$, with  
		$c_i, d \in \mathbb{N}$, $d \neq 0$, $B_i \in \cM_{n_i \times n}(\Z)$ and approximation parameter 
		$\epsilon > 0$. Assume that each entry of $B_i$ has absolute value at most $M$.
		Then Algorithm $\widetilde{G}$ runs in time $\poly(n, d, \log M, 1/\epsilon)$ 
		and computes $\BL(\mathbf{B, p})$ to a 
		multiplicative approximation factor of $(1\pm \epsilon)$.
\end{theorem}

\begin{proof}
	If $n \neq \dst\sum_{i=1}^m p_i n_i$ then we know that $\BL(\mathbf{B, p}) = \infty$. 
	Otherwise we correctly proceed.
	In step 2, we construct the operator $\widetilde{T}_{(\mathbf{B, p})}$ according to
	Construction~\ref{con:squaring-operator}. By item 1 of Lemma~\ref{lem:capacity-expansion}
	we have that $\capac(T_{(\mathbf{B, p})}) = \capac(\widetilde{T}_{(\mathbf{B, p})})^{1/nd}$.
	By Lemma~\ref{lem:main_reduction}, we have $\capac(T_{(\mathbf{B, p})}) = 1/\BL(\mathbf{B, p})^2$.
	Thus, 
	$$ \BL(\mathbf{B, p}) = (1/\capac(\widetilde{T}_{(\mathbf{B, p})}))^{1/2nd}. $$
	As step 3 returns $\alpha$ in the interval 
	$[(1- \eps) \cdot \capac(\widetilde{T}_{(\mathbf{B, p})}), \capac(\widetilde{T}_{(\mathbf{B, p})})]$
	we have that 
	$$ (1/\alpha)^{1/2nd} \in 
	[\capac(\widetilde{T}_{(\mathbf{B, p})})^{1/2nd}, 
	(1+\eps) \capac(\widetilde{T}_{(\mathbf{B, p})})^{1/2nd}] = 
	[\BL(\mathbf{B, p}), (1+\eps)\BL(\mathbf{B, p})]. $$ 
	Since algorithm $\wt{G}$ only uses once algorithm $G$ on
	an operator of dimension $nd \times nd$ with entries bounded by $M$, the running time of
	algorithm $\wt{G}$ is given by $\poly(n, d, \log(M), 1/\epsilon)$. 
	This concludes the analysis of algorithm $\wt{G}$. 
\end{proof}

\section{Continuity of the Brascamp-Lieb constant}\label{sec:cont_bl}

In this section, we note that the explicit bounds on the continuity of capacity of completely positive operators (at rational points) in \cite{GGOW} directly translate to bounds on the continuity of Brascamp-Lieb constants (at rational points). We will then discuss how this can be extended to irrational points.

The following theorem is proved in \cite{GGOW}. 

\begin{theorem}[Theorem 4.5 in \cite{GGOW}]\label{capacity_continuity}
Suppose $A_1,\ldots,A_m$ and $B_1,\ldots,B_m$ are two tuples of $n \times n$ matrices such that the bit-complexity of elements of $A_i$'s is $b$ and $||A_i - B_i|| \le \delta$ for all $i$. Let $T_A$ be the operator defined by $A_1,\ldots,A_m$ and $T_B$ be the operator defined by  $B_1,\ldots,B_m$. Then there exists a polynomial $P(n, b, \log(m))$ s.t. if $\delta \le \expon(-P(n,b, \log(m)))$ and $\capac(T_A) > 0$, then $\capac(T_B) > 0$. Furthermore
$$
\left(1 - \frac{P(n,b, \log(m))}{\log(1/\delta)^{1/3}} \right) \le \frac{\capac(T_A)}{\capac(T_B)} \le \left( 1 + \frac{P(n,b, \log(m))}{\log(1/\delta)^{1/3}} \right)
$$
\end{theorem}

From this we get the following theorem (continuity of BL constant at rational points) as an immediate corollary, due to Construction \ref{cons:main_reduction}.

\begin{theorem}[Theorem \ref{stability} restated]\label{BL_continuity}
Suppose $(\mathbf{B, p})$ and $(\mathbf{\widetilde{B},p})$ are two tuples of Brascamp-Lieb datum. Let $(n,n_1,\ldots,n_m)$ be the dimension vector and $p_1,\ldots,p_m = c_1/d, \ldots, c_m/d$. Also suppose the entries of $B_i$'s can be described using $b$ bits and $||B_i - \widetilde{B}_i|| \le \delta$ for all $i$. Then there exists a polynomial $P(n, d, b, \log(m))$ s.t. if $\delta \le \expon(-P(n, d, b, \log(m)))$ and $\BL(\mathbf{B, p})$ is finite, then $\BL(\mathbf{\widetilde{B},p})$ is finite. Furthermore
$$
\left(1 - \frac{P(n,d, b, \log(m))}{\log(1/\delta)^{1/3}} \right) \le \frac{\BL(\mathbf{B, p})}{\BL(\mathbf{\widetilde{B},p})} \le \left( 1 + \frac{P(n,d, b, \log(m))}{\log(1/\delta)^{1/3}} \right)
$$
\end{theorem}

\paragraph{\textit{Continuity at irrational points}}

We note that Theorem \ref{capacity_continuity} can be extended to continuity of capacity at irrational $A_1,\ldots,A_m$ (without any explicit bounds). We describe the changes that need to be made to the proof of Theorem 4.5 in \cite{GGOW}. 
\begin{enumerate}
\item First prove that for every $\eta > 0$, there is an $\eta$-approximate fixed point $C$ of the operator $T_A$ defined by $A_1,\ldots,A_m$ (if $\capac(T_A) > 0$). 
\item Using the fact that $||A_i-B_i|| \le \delta$, prove that $C$ is also an $\eta'$-approximate fixed point of the operator $T_B$ defined by  $B_1,\ldots,B_m$. $\delta$ can be chosen to be sufficiently small based on the lowest and highest eigenvalues of $C$. Now an application of Lemma 3.8 in \cite{GGOW} can finish the proof.
\end{enumerate}

When $\mathbf{p}$ is not rational, our reduction from BL datum to operators doesn't work anymore. Nonetheless, Algorithm \ref{The BL scaling algorithm} described in the introduction can be applied directly to the BL datum. Since this algorithm is continuous in $\mathbf{B}$, this should imply continuity of $\BL(\mathbf{B,p})$ in $\mathbf{B}$. However, it is not clear to us at the moment how long this algorithm will have to be run before the BL constant gets sufficiently close to $1$. This is because the progress per step in Theorem \ref{analysis} depends on the common denomintor $d$ of $\mathbf{p}$, which doesn't make sense in the case when $\mathbf{p}$ is not rational. But perhaps the fact that the vertices of the BL polytope $P_{\mathbf{B}}$ are rational can come in handy.



\section{Interesting special cases of BL polytopes} \label{sec:mat_int}

In this section, we explore special cases of BL polytopes that give rise to natural polytopes which are well studied in the combinatorial optimization literature. As mentioned before, Barthe \cite{Barthe2} proved that the BL polytopes corresponding to rank-$1$ BL datum are exactly the basis polytopes.

Now we show that some special cases of BL polytopes corresponding to rank-$2$ BL datum give rise to linear 
matroid intersection polytopes. Given a collection of vectors in $\mathbb{R}^n$, $\mathbf{v} = (v_1,\ldots,v_m)$, 
the linear matroid $\mathcal{M}_{\mathbf{v}}$ associated with this collection (over the ground set $[m]$) is the 
following set $\{I \subseteq [m]: (v_i)_{i \in I} \: \text{are linearly independent}\} $\footnote{Linear 
matroids can be defined w.r.t. any field $\mathbb{F}$ but we will restrict ourselves to $\mathbb{R}$ as this is 
the case which embeds into BL.}.

Given two linear matroids $\Matrv$ and $\Matrw$, their intersection polytope $P_{\Matrv, \Matrw}$ is the convex hull of the characteristic vectors of their common bases (sometimes this is defined as the convex hull of the common independent sets) i.e.
$$
P_{\Matrv, \Matrw} = \text{conv} \left\{ 1_I, I \subseteq [m]: (v_i)_{i \in I}, (w_i)_{i \in I} \: \: \text{both form a basis for } \mathbb{R}^n\right\}
$$

\noindent For a linear matroid $\Matrv$ and a set $J \subseteq [m]$, we define the subspace 
$$V_J  = \text{span}\left( (v_j)_{j \in J}\right)$$
It is well known (e.g. see Chapter 41 in \cite{schrijver_B}) $P_{\Matrv, \Matrw}$ can be described by the following linear program (with exponentially many inequalities!) on a variable vector  $\mathbf{p}$:
\begin{align}
&\sum_{j=1}^m p_j = n \nonumber \\
&\sum_{j \in J} p_j \le \text{dim} (V_J) \: \: \forall J \subseteq [m] \label{eqn:ankit4}\\
&\sum_{j \in J'} p_j \le \text{dim} (W_{J'}) \: \: \forall J' \subseteq [m] \label{eqn:ankit5}\\
&\mathbf{p} \ge 0 \nonumber
\end{align}

More generally, we can define equations~\eqref{eqn:ankit4} and~\eqref{eqn:ankit5} with respect to any subspace
$U$ of $\R^n$, so that it resembles the BL inequalities. By doing that we obtain:
\begin{align}
&\sum_{j =1}^m p_j \cdot 1\{ v_j \in U \} \le \text{dim} (U \cap V_{[m]}) \: \: \ \ \forall  
U \subseteq \R^n \label{eqn:ankit4-p}\\
&\sum_{j=1}^m p_j \cdot 1\{ w_j \in U \} \le \text{dim} (U \cap W_{[m]}) \: \: \forall  
U \subseteq \R^n \label{eqn:ankit5-p}
\end{align}


Now consider the following BL datum: for every $j\in [m]$, $B_j : \mathbb{R}^{2n} \rightarrow \mathbb{R}^2$ 
is given by 
$$
B_j(x,y) = (\langle v_j,x \rangle, \langle w_j, y \rangle)
$$
We will show that the BL polytope $P_{\mathbf{B}}$ corresponding to this BL datum is the 
same as the matroid intersection polytope $ P_{\Matrv, \Matrw}$. Our feasibility algorithm for BL-polytopes 
thus automatically gives a polynomial time algorithm for the linear matroid intersection problem 
(here over $\mathbb{R}$). Of course it was already known that this problem is in $P$, and moreover 
Gurvits~\cite{gurvits2004} already noticed that operator scaling efficiently solves this problem; our hope is that 
such encodings are possible for other optimization problems with exponential size linear programs.

\begin{lemma}\label{lem:BL-eq-mat} 
The BL polytope $P_{\mathbf{B}}$ corresponding to the above BL datum 
is the same as $P_{\Matrv, \Matrw}$.
\end{lemma}

\begin{proof}

Recall that by Theorem~\ref{feasibility}, a point $\mathbf{p} \in P_{\mathbf{B}}$ is given by the following constraints:
\begin{align}
&\sum_{j=1}^m p_j \cdot 2 = 2n \nonumber \\
&\sum_{j=1}^m p_j \text{dim}(B_j(U)) \ge \text{dim}(U) \: \: \forall \: \: \text{subspaces } U \text{ of } \mathbb{R}^{2n} \label{eqn:ankit3}\\
&\mathbf{p} \ge 0 \nonumber
\end{align}

To show that $P_\mathbf{B} \subseteq P_{\Matrv, \Matrw}$, it is enough to show that the set of 
equations~\eqref{eqn:ankit3} generate the set of equations~\eqref{eqn:ankit4} and~\eqref{eqn:ankit5}.
We can do this as follows: let 
$\tilde{V}_J = \{ (v_j, 0) \in \R^{2n} : v_j \in V_J \}$ (or 
$\tilde{W}_{J'} =\{ (0, w_j) \in \R^{2n} : w_j \in W_{J'} \}$). In~\eqref{eqn:ankit3}, setting 
$U$ to be the space $\tilde{V}_J^\perp$ (or $\tilde{W}_{J'}^\perp$), that is,
the space of vectors orthogonal to $\tilde{V}_J$ (or $\tilde{W}_{J'}$) we obtain:
\begin{align*}
	\sum_{j=1}^m p_j \cdot \text{dim}(B_j(\tilde{V}_J^\perp)) \ge \text{dim}(\tilde{V}_J^\perp) 
	\ &\iff \ \sum_{j=1}^m p_j \cdot (2 - 1\{ v_j \in V_J \}) \ge 2n - \text{dim}(V_J) \\
	&\iff \ \sum_{j=1}^m p_j \cdot 1\{ v_j \in V_J \} \leq \text{dim}(V_J).  
\end{align*} 
We have thus generated equation~\eqref{eqn:ankit4}. To generate equation~\eqref{eqn:ankit5} we only need
to repeat the procedure above with $U = \tilde{W}_{J'}^\perp$. This completes this part of the proof.

Now, we need to show that $P_{\Matrv, \Matrw} \subseteq P_\mathbf{B}$. To prove this, 
it is enough to show that equations~\eqref{eqn:ankit4-p} and~\eqref{eqn:ankit5-p} 
generate the set of equations~\eqref{eqn:ankit3}.
Given a subspace $U \subseteq \mathbb{R}^{2n}$, let $U_1 = \pi_1(U) \subseteq \R^n$ 
($U_2 = \pi_2(U)$) be the projection of $U$ onto the first (second) half of $\R^{2n}$.
Additionally, let  
\begin{align*}
	\tilde{Z}_1 &= \{ (u, 0) \in U \ : \ u \in \R^n \} \\
	\tilde{Z}_2 &= \{ (0, u) \in U \ : \ u \in \R^n \}
\end{align*}
that is, $\tilde{Z}_1$ (resp. $\tilde{Z}_2$) are the vectors of $U$ whose second 
(resp. first) half of coordinates being zero. Let $Z_1 = \pi_1(\tilde{Z}_1)$ and
$Z_2 = \pi_2(\tilde{Z}_2)$. We can decompose $U$ as a direct sum 
as follows: $U = \tilde{Z}_1 \oplus \tilde{Z}_2 \oplus Z$, where $Z$ is orthogonal to $\tilde{Z}_1 \oplus \tilde{Z}_2$. 
Note that $\dm(Z) = \dm(\pi_1(Z)) = \dm(\pi_2(Z))$, as
any basis $\{z_1, \ldots, z_d\}$ of $Z$ yields a basis 
$\{\pi_j(z_1), \ldots, \pi_j(z_d)\}$ for $\pi_j(Z)$, where $j \in \{1,2\}$. 
This is because for $z \in Z$, $\pi_j(z) = 0$ implies $z = 0$ since $Z$ is orthogonal to $\tilde{Z}_1 \oplus \tilde{Z}_2$.
Since $U = \tilde{Z}_1 \oplus \tilde{Z}_2 \oplus Z$, we also have that
$\pi_j(U) = Z_j \oplus \pi_j(Z)$, for $j \in \{1,2\}$.

From $U = \tilde{Z}_1 \oplus \tilde{Z}_2 \oplus Z$ and the properties above we get  
\begin{align}\label{eq:dim-eq}
	2\cdot\dm(U) &= 2 \cdot (\dm(Z_1) + \dm(Z_2) + \dm(Z))  \nonumber \\
	&= \dm(Z_1) + \dm(Z_2) + (\dm(Z_1) + \dm(\pi_1(Z))) + (\dm(Z_2) + \dm(\pi_2(Z)))
	\nonumber \\
	&= \dm(Z_1) + \dm(Z_2) + \dim(U_1) + \dm(U_2)
\end{align}

Letting $U_1^\perp, U_2^\perp, Z_1^\perp, Z_2^\perp$ denote the orthogonal complements 
of $U_1, U_2, Z_1, Z_2$ respectively (inside $\mathbb{R}^n$) we have that the following
inequality always holds:
\begin{equation}\label{eq:dim-ineq}
	\dm(B_j(U)) \geq \dfrac{1}{2} \cdot 
	(4 - 1\{v_j \in U_1^{\bot}\} - 1\{v_j \in Z_1^{\bot}\} 
	- 1\{w_j \in U_2^{\bot}\} - 1\{w_j \in Z_2^{\bot}\})  
\end{equation}

The proof follows from the case analysis below:

\textbf{Case 1:} $\dm(B_j(U)) = 2$. In this case, we are done, as RHS is always
no greater than 2.

\textbf{Case 2:} $\dm(B_j(U)) = 1$. Here, we have two subcases. 
\begin{enumerate}
	\item If $v_j \not\in Z_1^\perp$ (or similarly $w_j \not\in Z_2^\perp$), then 
	it must be the case that $w_j \in U_2^\perp$ (because of $\dm(B_j(U)) = 1$), which implies that 
	$w_j \in Z_2^\perp$, as $U_2^\perp \subseteq Z_2^\perp$. This implies that the 
	RHS is $\leq 1$.
	\item If $v_j \in Z_1^\perp$ and $w_j \in Z_2^\perp$ we have that the RHS is
	already $\leq 1$.
\end{enumerate}

\textbf{Case 3:} $\dm(B_j(U)) = 0$. In this case, we must have 
$v_j \in U_1^\perp \subseteq Z_1^\perp$ and $w_j \in U_2^\perp \subseteq Z_2^\perp$.
The RHS then becomes zero.

Now, if we apply equations~\eqref{eqn:ankit4-p} and~\eqref{eqn:ankit5-p} to the
spaces $U_1^\perp, U_2^\perp, Z_1^\perp, Z_2^\perp$, we obtain:

\begin{align*}
	\sum_{j =1}^m p_j \cdot 1\{ v_j \in U_1^\perp \} &\le \dm(U_1^\perp \cap V_{[m]}) \leq \dm(U_1^\perp) \\
	\sum_{j =1}^m p_j \cdot 1\{ v_j \in Z_1^\perp \} &\le \dm(Z_1^\perp \cap V_{[m]}) \leq \dm(Z_1^\perp) \\
	\sum_{j =1}^m p_j \cdot 1\{ w_j \in U_2^\perp \} &\le \dm(U_2^\perp \cap W_{[m]}) \leq \dm(U_2^\perp) \\
	\sum_{j =1}^m p_j \cdot 1\{ w_j \in Z_2^\perp \} &\le \dm(Z_2^\perp \cap W_{[m]}) \leq \dm(Z_2^\perp) 
\end{align*}

Adding the equations above, we obtain:

\begin{align}
&\: \sum_{j =1}^m p_j \cdot (1\{ v_j \in U_1^\perp \} + 1\{ v_j \in Z_1^\perp \} +
1\{ w_j \in U_2^\perp \} + 1\{ w_j \in Z_2^\perp \}) \nonumber \\
&\leq 
\dm(U_1^\perp) +\dm(Z_1^\perp) + \dm(U_2^\perp) + \dm(Z_2^\perp) \label{eq:bl-pol-ineq}
\end{align}

By inequality~\eqref{eq:dim-ineq}, we have that 
$$
1\{ v_j \in U_1^\perp \} + 1\{ v_j \in Z_1^\perp \} +
1\{ w_j \in U_2^\perp \} + 1\{ w_j \in Z_2^\perp \} 
\geq 
4 - 2 \cdot \dm(B_j (U))
$$
And by equation~\eqref{eq:dim-eq}, we have that 
\begin{align*}
\dm(U_1^\perp) +\dm(Z_1^\perp) + \dm(U_2^\perp) + \dm(Z_2^\perp)
&= 4n - \dm(U_1) +\dm(Z_1) + \dm(U_2) + \dm(Z_2) \\
&= 4n - 2\cdot\dm(U)
\end{align*}

Putting these facts altogether, inequality~\eqref{eq:bl-pol-ineq} implies:
$$
\sum_{j =1}^m p_j \cdot (4 - 2\cdot\dm(B_j(U))) \leq 4n - 2\cdot\dm(U)
$$
which is equivalent to $\sum_{j =1}^m p_j \cdot \dm(B_j(U)) \geq \dm(U)$, as we wanted.
\end{proof}

Since the perfect bipartite matching polytopes are a special case of linear 
matroid intersection polytopes, we see that the bipartite matching polytopes are 
a special case of BL polytopes for rank-$2$ BL datum. 
As mentioned before, it is an interesting question if the general perfect 
matching polytopes are a special case of BL polytopes or not. 

\section{Additional properties of the Brascamp-Lieb constant}\label{sec:properties}

In this section, we prove that under appropriate normalization conditions on the BL datum, the BL constant is at least $1$, with equality iff the datum is geometric. This is a strengthening of the following theorem proved in \cite{Vald11}. We will assume that exponent vectors $\mathbf{p}$ satisfy $p_i > 0$ for all $i$.

\begin{theorem}[Theorem 1,5 of \cite{Vald11}] Let $(\mathbf{B,p})$ be a projection-normalized and feasible BL datum. Then $\BL(\mathbf{B,p}) \ge 1$ and equality holds iff the datum is geometric.
\end{theorem}

We prove the following strengthening of the above theorem. 

\begin{theorem}\label{thm:BLge1} Let $(\mathbf{B,p})$ be a feasible BL datum that satisfies 
$$
\sum_{i=1}^m p_i \tr\left[ B_i^{\dagger} B_i\right] = n
$$
Then $\BL(\mathbf{B,p}) \ge 1$ and equality holds iff the datum is geometric.
\end{theorem}

It is easy to see that we get the following corollary from the above theorem.

\begin{corollary}\label{cor:isotropy_or_projection} Let $(\mathbf{B,p})$ be a feasible BL datum that is either projection-normalized or isotropy normalized. Then $\BL(\mathbf{B,p}) \ge 1$ and equality holds iff the datum is geometric.
\end{corollary}

\begin{proof} We will apply Theorem \ref{thm:BLge1}. The only thing that needs to be verified is that when $(\mathbf{B,p})$ is projection-normalized or isotropy normalized, then 
$$
\sum_{i=1}^m p_i \tr\left[ B_i^{\dagger} B_i\right] = n
$$
When $(\mathbf{B,p})$ is projection-normalized, then $\tr\left[ B_i^{\dagger} B_i\right] = \tr\left[B_i B_i^{\dagger}\right] = \tr[I_{n_i}] = n_i$ and hence
$$
\sum_{i=1}^m p_i \tr\left[ B_i^{\dagger} B_i\right] = \sum_{i=1}^m p_i n_i = n
$$
Note that $\sum_{i=1}^m p_i n_i = n$ is satisfied since $(\mathbf{B,p})$ is feasible. When $(\mathbf{B,p})$ is isotropy-normalized, then
$$
\sum_{i=1}^m p_i \tr\left[ B_i^{\dagger} B_i\right] = \tr\left[\sum_{i=1}^m p_i  B_i^{\dagger} B_i\right] = \tr[I_n] = n
$$
This completes the proof.
\end{proof}

For rational exponent vectors $\mathbf{p}$, Theorem \ref{thm:BLge1} would follow from the reduction to operator scaling and invoking appropriate theorems from \cite{GGOW}. However we don't assume rationality of $\mathbf{p}$ in this section. It turns out that one can imitate the proofs of corresponding statements in the operator scaling world\footnote{Let $T$ be a positive operator acting on $n \times n$ matrices s.t. $\tr[T(I)] = n$. Then $\capac(T) \le 1$ and equality holds iff $T$ is doubly stochastic.} and the proofs work even in the irrational case. For the proof, we would need the following easy proposition, which is a consequence of the AM-GM inequality applied to the eigenvalues of $A$.

\begin{proposition}\label{prop:AM-GM_psd} Let $A \in \mathcal{M}_n(\C)$ be an $n \times n$ Hermitian positive semidefinite matrix s.t. $\tr[A] = n$. Then $\Det(A) \le 1$ and equality holds iff $A = I_n$. 
\end{proposition}

Now we are ready to prove our theorem. 

\begin{proof}(Of Theorem \ref{thm:BLge1}) Recall that
$$
1/\BL(\mathbf{B,p})^2 =  \inf \frac{\Det \left(\sum_{i=1}^m p_i B_i^{\dagger} X_i B_i \right)} {\prod_{i=1}^m \Det (X_i)^{p_i}}
$$
where the infimum is taken over all choices of positive definite matrices $X_i$ in dimension $n_i$. By plugging in $X_i = I_{n_i}$, we can see that 
\begin{align*}
1/\BL(\mathbf{B,p})^2 \le \Det \left(\sum_{i=1}^m p_i B_i^{\dagger} B_i \right)
\le 1
\end{align*}
Here the second inequality follows from Proposition \ref{prop:AM-GM_psd} and the fact that $\tr\left[ \sum_{i=1}^m p_i B_i^{\dagger} B_i\right] = n$, which follows from the assumption in the theorem statement. We also get from Proposition \ref{prop:AM-GM_psd} that equality holds only if $\sum_{i=1}^m p_i B_i^{\dagger} B_i = I_n$ i.e. $(\mathbf{B,p})$ is isotropy-normalized.

We also know that for a geometric datum the BL constant is $1$. So the only thing that is left is to show that if the BL constant is $1$, then $(\mathbf{B,p})$ is projection-normalized as well. For this, we plug in $X_i = \left( B_i B_i^{\dagger}\right)^{-1}$ (note that $B_i B_i^{\dagger}$ is invertible since $(\mathbf{B,p})$ is feasible) to get:
\begin{align*}
1/\BL(\mathbf{B,p})^2 \le \Det \left(\sum_{i=1}^m p_i B_i^{\dagger} \left(B_i B_i^{\dagger} \right)^{-1} B_i \right) \cdot \prod_{i=1}^m \Det \left(B_i B_i^{\dagger}\right)^{p_i}
\end{align*}
First note that
\begin{align*}
\tr\left[\sum_{i=1}^m p_i B_i^{\dagger} \left(B_i B_i^{\dagger} \right)^{-1} B_i \right] &= \sum_{i=1}^m p_i \tr\left[B_i^{\dagger} \left(B_i B_i^{\dagger} \right)^{-1} B_i \right] \\
&= \sum_{i=1}^m p_i \tr\left[ \left(B_i B_i^{\dagger} \right)^{-1} B_i B_i^{\dagger}\right] \\
&= \sum_{i=1}^m p_i \tr[I_{n_i}] = \sum_{i=1}^m p_i n_i = n
\end{align*}
Hence, by Proposition \ref{prop:AM-GM_psd}, we get that 
$$
\Det \left(\sum_{i=1}^m p_i B_i^{\dagger} \left(B_i B_i^{\dagger} \right)^{-1} B_i \right) \le 1
$$
and thus
\begin{align*}
1/\BL(\mathbf{B,p})^2 \le \prod_{i=1}^m \Det \left(B_i B_i^{\dagger}\right)^{p_i}
\end{align*}
Let $\lambda_{i,1},\ldots, \lambda_{i,n_i}$ be the eigenvalues of $B_i B_i^{\dagger}$. Then
\begin{align*}
 \log \left( \prod_{i=1}^m \Det \left(B_i B_i^{\dagger}\right)^{p_i} \right) &= \sum_{i=1}^m p_i \sum_{j=1}^{n_i} \log\left(\lambda_{i,j} \right) \\
&= \sum_{i=1}^m p_i n_i \cdot \frac{1}{n_i} \sum_{j=1}^{n_i} \log\left(\lambda_{i,j} \right) \\
& \le \sum_{i=1}^m p_i n_i \log\left(\frac{1}{n_i} \sum_{j=1}^{n_i} \lambda_{i,j}\right) \\
&= n \cdot  \sum_{i=1}^m \frac{p_i n_i}{n} \log\left(\tr\left[B_i B_i^{\dagger} \right]/n_i\right) \\
&\le n \cdot \log\left( \sum_{i=1}^m \frac{p_i}{n} \cdot \tr\left[B_i B_i^{\dagger} \right] \right) \\
&= 0
\end{align*}
The two inequalities follow from concavity of log. The last equality follows from the assumption that 
$$
\sum_{i=1}^m p_i \tr\left[ B_i^{\dagger} B_i\right] = n
$$
In the first inequality, equality holds iff for all $i$, $\lambda_{i,j}$'s are equal for all $j$ i.e. for all $i$, $B_i B_i^{\dagger} = c_i I_{n_i}$ for some $c_i > 0$. In the second inequality, equality holds iff $\tr\left[B_i B_i^{\dagger} \right]/n_i$ is equal for all $i$ i.e. $c_i = c$ for some $c > 0$. By the normalization condition
$$
\sum_{i=1}^m p_i \tr\left[ B_i^{\dagger} B_i\right] = n
$$
we get that $c = 1$ and hence $B_i B_i^{\dagger} = I_{n_i}$ for all $i$ i.e. $(\mathbf{B,p})$ is isotropy normalized. This completes the proof.
\end{proof}

\section{Iteration analysis for Algorithm \ref{The BL scaling algorithm}}\label{sec:irrational}

In this section, we give an analysis for the number of iterations required in Algorithm \ref{The BL scaling algorithm} to get close to the geometric position. While for the rational exponent case, such an analysis follows from the reduction to operators, here we give an analysis that works even in the irrational case. The proof closely mimicks the argument of \cite{GGOW} in the operator world.

Let us first define the analogue of the ds measure for BL datum which measures how close the BL datum is to geometric position.

\begin{definition}[\textbf{Distance to geometric position}] 
$$
g(\mathbf{B,p}) = \tr\left[ \left( \sum_{i=1}^m p_i B_i^{\dagger} B_i - I_n\right)^2\right] + \sum_{i=1}^n \tr \left[ \left(B_i B_i^{\dagger} - I_{n_i} \right)^2\right]
$$
\end{definition}

Note that $(\mathbf{B,p})$ is geometric iff $g(\mathbf{B,p})  = 0$. We will need the following lemma which is a robust version of the AM-GM inequality and is a consequence of Lemma 3.10 in \cite{LSW}. 

\begin{lemma}\label{lem:AM-GM_robust} Let $A \in \cM_n(\mathbb{C})$ be a Hermitian positive semidefinite matrix s.t. $\tr[A] = n$ and $\tr\left[ \left( A-I_n\right)^2\right] \ge \eps$, where $\eps \le 1$. Then $\Det(A) \le \expon(-\eps/6)$.
\end{lemma}

We are now ready to state the bound on the number of iterations.

\begin{theorem} Let $(\mathbf{B,p})$ is a feasible datum that is isotropy or projection normalized s.t. $B_i \in \cM_{n_i \times n}(\mathbb{Q})$ and each entry of $B_i$ has bit-size atmost $b$. Then for any $\eps \le 1$ and $t= O\left( \frac{nm \log(nm) + nb}{\eps}\right)$, the sequence $\mathbf{B^0} = \mathbf{B}, \mathbf{B^1},\ldots, \mathbf{B^t}$ in Algorithm \ref{The BL scaling algorithm} satisfies $\min_{i \in \{0,1,\ldots,t\}} g\left( \mathbf{B^i, p}\right) \le \eps$.
\end{theorem}

We remark that the assumption on the datum being normalized is without loss of generality as it can be achieved in one extra iteration. The proof will be along the lines described in the introduction and will proceed by analyzing the effect of normalization operations on the BL constant. 

\begin{proof}
By Corollaries \ref{cor:bl_upperbound_ind} and \ref{cor:bl_upperbound_rational}, we have that $\BL(\mathbf{B,p}) \le \expon(12nm \log(nm) + nb)$. Assume on the contrary that $\min_{i \in \{0,1,\ldots,t\}} g\left( \mathbf{B^i, p}\right) > \eps$. We will prove that an isotropy normalization step decreases the BL constant by a multiplicative factor of $\expon(-\eps/12)$ and a projection normalization step decreases the BL constant\footnote{It is also possible to get a quantitative bound on the decrease in the projection step but it is not needed.}. We also know by Corollary \ref{cor:isotropy_or_projection} that the BL constant is at least $1$ in the isotropy or projection normalized case. This will complete the contradiction and the proof.

Let us first analyze the isotropy normalization step. Let $(\mathbf{B',p})$ be the datum before the normalizaton and $(\mathbf{B'',p})$ after. Since $(\mathbf{B',p})$ is projection normalized, we have that 
$$
\tr\left[ \left( \sum_{i=1}^m p_i B_i^{'\dagger} B'_i - I_n\right)^2\right] = g(\mathbf{B',p}) > \eps
$$
We also know that
$$
\BL(\mathbf{B'',p}) = \Det\left( \sum_{i=1}^m p_i B_i^{'\dagger} B'_i\right)^{1/2} \cdot \BL(\mathbf{B',p})
$$
Hence we can use Lemma \ref{lem:AM-GM_robust} to conclude that the BL constant decreases by a multiplicative factor of $\expon(-\eps/12)$. 

Now let us look at the projection normalization step. Let $(\mathbf{B',p})$ be the datum before the normalizaton and $(\mathbf{B'',p})$ after. We know that 
$$
\BL(\mathbf{B'',p}) = \prod_{i=1}^m \Det\left( B_i' B_i^{'\dagger}\right)^{p_i/2} \cdot \BL(\mathbf{B',p})
$$
Since $(\mathbf{B',p})$ is isotropy normalized, we can conclude that $ \prod_{i=1}^m \Det\left( B_i' B_i^{'\dagger}\right)^{p_i/2} \le 1$ by the same log-concavity argument as in the proof of Theorem \ref{thm:BLge1}. This concludes the proof.
\end{proof}

\section{Conclusion and open questions}\label{sec:open}

In this paper we address some of the central algorithmic questions regarding Brascamp-Lieb inequalities. 
In particular, we gave polynomial time algorithms to test membership in the BL polytopes, compute a violating inequality (separation oracle) and approximate the BL constants. These algorithms were obtained by a simple reduction from the BL data to completely positive operators and then using recent polynomial time algorithms that test whether such operators are rank decreasing and compute the capacity of such operators \cite{GGOW,IQS15b}. The reduction itself provides in some cases different proofs of some of the known structural results regarding BL inequalities, in particular the characterization of feasibility.

On the algorithmic side, our work leaves some concrete open questions.
\begin{enumerate}
\item While our reduction is polynomial time in the binary description length of the matrices $(B_1,\ldots,B_m)$ in the BL datum, it is not polynomial time in the binary description length of the exponent vector $p_1,\ldots,p_m$. When $p_1,\ldots,p_m = \frac{c_1}{d},\ldots,\frac{c_m}{d}$, the running times of our algorithms scale as $\poly(d)$. Can one improve them to get (the optimal) $\poly(\log(d))$?
\item Our algorithm for approximating the BL constants up to an $(1+\eps)$ factor run in time $\poly(1/\eps)$.
Can one improve it to a running time that scales as $\poly(\log(1/\eps))$?

\item We provide a weak separation oracle for the BL polytope. Sufficiently strong separation oracles for convex programs yield optimization algorithms via the general reduction of~\cite{GLS88}. However, our current separation oracle does not seem strong enough. Can one efficiently optimize linear functions over BL polytopes?
\end{enumerate}

\noindent On the structural side, our reduction may provide a higher level view of BL inequalities, simply since BL data is just a proper subset of completely positive operators, and viewing them as such may lead to better understanding. One intriguing question is whether this viewpoint leads to an even larger family of inequalities. It is obvious that for any doubly stochastic operator $T$, the fact that $\capac(T) = 1$ implies the following determinantal inequality (that holds for all psd matrices $X$), 
$$
\Det\left( \sum_{i=1}^m A_i X A_i^{\dagger} \right) = \Det(T(X)) \ge \Det(X)
$$
which can be translated to an inequality on Gaussian densities.
One can wonder for which operators $T$ the Gaussian densities can be replaced by arbitrary non-negative densities, as in BL-inequalities.
There are simple counterexamples showing this does not hold in full generality. 
On the other other hand, our reduction and the connection to quiver representation (Section \ref{sec:quiver}) tells us that 
 for operators $T$ arising from the star quiver, these inequalities are in $1$-$1$ correspondence to the BL-inequalities.
 It would be interesting to better understand what is special about star quivers in this context. In particular, do other quivers 
 lead to new inequalities?

\section*{Acknowledgments} 
We would like to thank Zeev Dvir for pointing us to the paper \cite{BCCT}. We would 
also like to thank Jonathan Bennett, Larry Guth and anonymous reviewers for helpful 
comments. 

We thank Rohit Gurjar for suggesting that BL polytopes can encode the linear 
matroid intersection polytopes, and we would like to thank Nisheeth Vishnoi and 
Damian Straszak for pointing out to us a mistake in our previous proof of 
Lemma~\ref{lem:BL-eq-mat} and for providing us a fix for it. 

\bibliographystyle{alpha}
\bibliography{refs}

\end{document}